\documentclass[aps,pra,reprint,superscriptaddress,nofootinbib]{revtex4-2}

\pdfoutput=1

\usepackage[usenames]{xcolor}
\usepackage{graphicx}
\usepackage[caption=false,singlelinecheck=off]{subfig} 

\usepackage{tikz}
\usetikzlibrary{cd}

\usepackage{dsfont}

\usepackage{amsmath, amssymb, amsfonts, amsthm}
\usepackage{units}
\usepackage{mathtools}
\usepackage{thmtools, thm-restate} 

\newcommand{\ket}[1]{|  #1 \rangle}
\newcommand{\bra}[1]{\langle #1 |}

\newcommand{\hilbert}{\mathcal{H}}

\newcommand{\C}{\mathbb{C}}
\newcommand{\Z}{\mathbb{Z}}
\newcommand{\N}{\mathbb{N}}

\newcommand{\im}{\mathrm{im}} 

\newcommand{\tr}{\operatorname{Tr}  }

\declaretheorem[name={Remark}, style=remark, unnumbered]{remark}

\newtheorem{theorem}{Theorem}[]

\newtheorem{proposition}[theorem]{Proposition}
\newtheorem{definition}[theorem]{Definition}

\newtheorem{problem}[theorem]{Problem}
\newtheorem{promise}[theorem]{Promise}

\begin{document}


\title{Thermodynamic optimization of quantum algorithms:  On-the-go erasure of qubit registers}


\author{Florian Meier}
\email[]{florian.meier@tuwien.ac.at}
\affiliation{Institute for Theoretical Physics, ETH Zurich, 8093 Z\"{u}rich, Switzerland}
\affiliation{Atominstitut, Technische Universität Wien, 1020 Vienna, Austria}
\author{Lídia del Rio}
\email[]{lidia@phys.ethz.ch}
\affiliation{Institute for Theoretical Physics, ETH Zurich, 8093 Z\"{u}rich, Switzerland}


\date{\today}

\begin{abstract}
We consider two bottlenecks in quantum computing: limited memory size and noise caused by heat dissipation. Trying to optimize both, we investigate ``on-the-go erasure'' of quantum registers that are no longer needed for a given algorithm: freeing up auxiliary qubits as they stop being useful would facilitate the parallelization of computations. We study the minimal thermodynamic cost of erasure in these scenarios, applying results on the Landauer erasure of entangled quantum registers. For the class of algorithms solving the Abelian hidden subgroup problem, we find optimal on-the-go erasure protocols. We conclude that there is a trade-off: if we have enough partial information about a problem to build efficient on-the-go erasure, we can use it to instead simplify the algorithm, so that fewer qubits are needed to run the computation in the first place. We provide explicit protocols for these two approaches.
\end{abstract}


\maketitle


\section{Introduction}
\label{sec:introduction}
\paragraph*{When is the best time to reset qubit registers?} A default option is to run a whole algorithm and reset all registers to $\ket0$ at the end, after the final measurements. However, if the total number of qubits is a limitation and we need to run several algorithms concurrently, we may want to free up some registers as they stop being useful: for example, in the period-finding algorithm, the auxiliary register can be discarded after applying the oracle (Figure~\ref{fig:algorithm_vanilla}).
Another critical factor may be heat dissipation: Landauer's principle tells us that the erasure of every single qubit from a fully mixed state to $\ket0$ has a fundamental work cost of $k_B T \ln 2$ if performed at temperature $T$, releasing the same amount of heat to the environment \cite{landauer}. As heat dissipation in a quantum computer threatens coherence, reducing the work cost of erasure may be of critical importance.

\begin{figure}
    \subfloat[\label{fig:algorithm_vanilla} Class of algorithms considered.]{\includegraphics[width=\columnwidth]{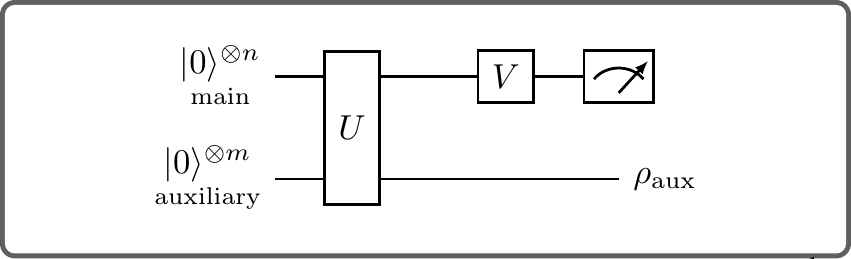}}\\
    \subfloat[\label{fig:algorithm_otg_bruteforce} On-the-go brute-force erasure.]{\includegraphics[width=\columnwidth]{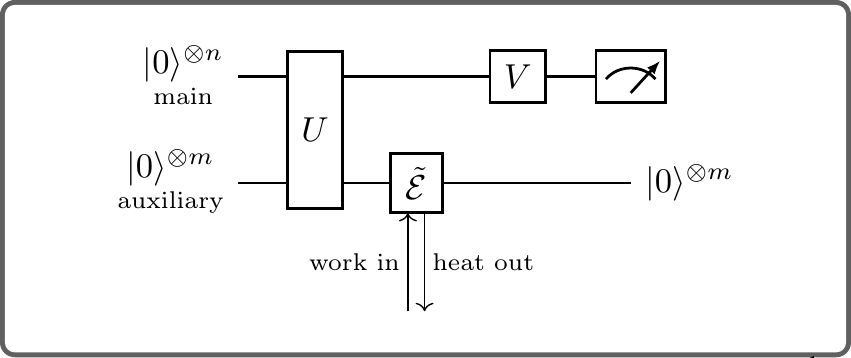}}\\
    \subfloat[\label{fig:algorithm_bennet_uncompute} Bennett's uncomputing erasure.]{\includegraphics[width=\columnwidth]{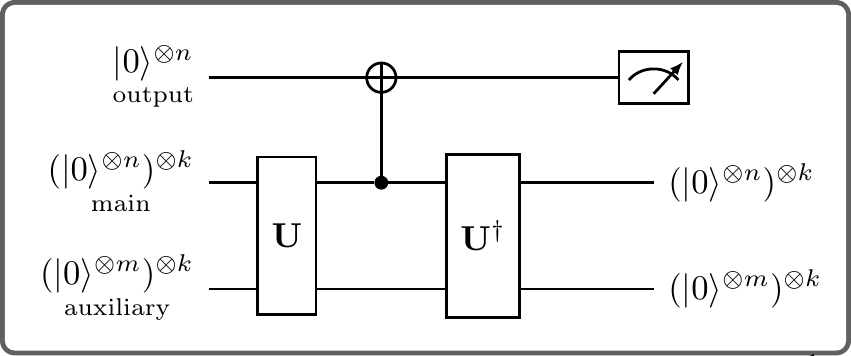}}\\
    \subfloat[\label{fig:algorithm_otg_optimized} Optimized on-the-go erasure.]{\includegraphics[width=\columnwidth]{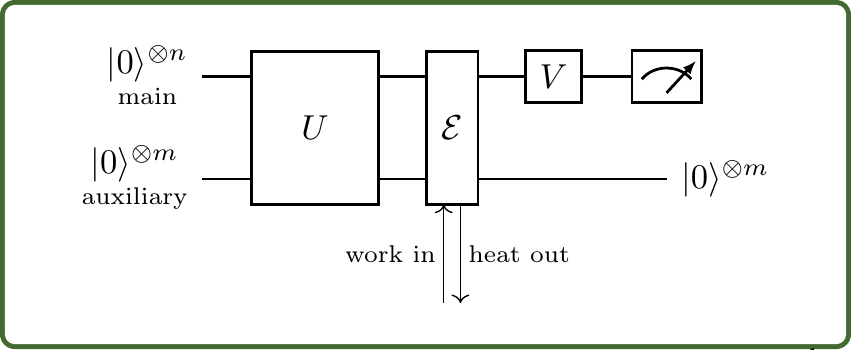}}
    \caption{\textit{Qubit erasures.} Green frames are new proposals.
    (\ref{fig:algorithm_vanilla}) We consider quantum algorithms where a main register is used until the final measurement, but there can be auxiliary registers, which are only needed for part of the algorithm --- for example, the period-finding algorithm and more generally algorithms for hidden subgroup problems are of this form. 
    (\ref{fig:algorithm_otg_bruteforce}) In order to free up memory space, the auxiliary register can be erased on-the-go. A brute-force erasure at temperature $T$ will have work cost $k_BT\ln2$ per qubit due to Landauer's principle.
    (\ref{fig:algorithm_bennet_uncompute}) For Bennett's uncomputing, the original probabilistic algorithm is made essentially deterministic by running many copies of $(V\otimes\mathds{1})\circ U$ in parallel together with a quantum implementation of the classical post-processing, summarized as $\mathbf{U}$. The result of this calculation is coherently copied to an output register and $\mathbf{U}$ is uncomputed.
    (\ref{fig:algorithm_otg_optimized}) We propose an optimized on-the-go erasure scheme which takes advantage of the entanglement between main and auxiliary register to reduce the work cost of erasure of the latter.}
    \label{fig:comparing_erasures}
\end{figure}

\paragraph*{Previous erasure schemes.} We consider algorithms that use a \emph{main register} of $n$ qubits and an \emph{auxiliary register} of $m$ qubits (Figure~\ref{fig:algorithm_vanilla}); the latter can be discarded at some halfway point in the algorithm. For example, the period-finding algorithm is of this form. To optimize memory space, we may want to erase it as soon as possible: a brute-force erasure procedure  of those $m$ qubits (Figure~\ref{fig:algorithm_otg_bruteforce}) would dissipate heat $m\ k_B T \ln2$ with a simple fixed map, independent of the algorithm.
On the other extreme, if we only want to optimize the heat cost, we can apply Bennett's reversible erasure procedure \cite{Bennett1982,watrous2008quantum,Baumeler2019}, which coherently copies the output register to an external system, and then uncomputes the algorithm's circuit on the original qubits reversibly (Figure \ref{fig:algorithm_bennet_uncompute}). The main drawback of this procedure emerges when it is applied to probabilistic quantum algorithms like period-finding: Bennett's uncomputing only works when the output register is decoupled from the rest of the quantum computer --- in other words, when the algorithm outputs a deterministic result  \cite{Rio2011a}. Probabilistic quantum algorithms are made (approximately) deterministic by repeating them many times, and applying classical post-processing to the probabilistic outputs, including for example a majority vote. To apply Bennett's uncomputing to a probabilistic algorithm like period-finding, we would have to implement all these runs of the algorithm and the (usually classical) post-processing as a large reversible quantum circuit, so that the final post-processed quantum output is approximately decoupled from the rest of the memory.%
\footnote{For example, majority votes can  be implemented reversibly \cite{Nashiry2017}.} 
This process has a large complexity cost both in terms of memory size and circuit length; while it may be worth pursuing in a distant future when our computers work flawlessly and reversibly at the quantum level,  in this work we focus on a NISQ regime, and try  to optimize both thermodynamic and computational complexity costs of algorithms.  

\paragraph*{Thermodynamic considerations.}  We will work in the quantum resource theory of thermal operations \cite{2019review_lostaglio,Brandao2013,Liu2019}. In this framework, unitary operations on degenerate systems are given for free, and irreversible operations like erasure have associated work costs. 
Contemporary quantum computers are of course still far from this ideal scenario; nonetheless, the fundamental limits for the energy cost of implementing single-qubit unitaries are comparable to that of erasure \cite{CG21}. 
Moreover, note that the energy requirements to implement common unitary operations (which depend on the quantum control mechanisms) scale  sublinearly on the number of qubits, while erasure scales linearly \cite{CG21}. This, together with recent erasure experiments that approach Landauer's limit  \cite{Orlov2012,KJ14,CS21}, have us speculate that energies of the order $k_B T$ may eventually become relevant to quantum computing.
Overall, thermodynamic optimization of quantum computation entails at least three independent components: (1) cost of unitary gates, (2) cost of erasure of fully mixed qubits, (3) optimizing number of fully mixed qubits that must be erased (see \cite{Goold2016,Binder2018,Taranto2021,Auffeves2022} for reviews on the thermodynamics of quantum computation). Our work addresses the third component, and can be applied in conjunction with restrictions or improvements on the former two. This is further discussed in Section~\ref{sec:discussion}.

\subsection{Contribution of this paper} 
\paragraph*{Optimized on-the-go erasure.} Making use of entanglement between the main and auxiliary register as a thermodynamic resource \cite{Rio_2011,Goold2016,2019review_lostaglio,Brandao2013,Brandao2015,Liu2019}, we introduce a new erasure scheme (Figure~\ref{fig:algorithm_otg_optimized}).
It entails a strictly lower heat dissipation than brute-force erasure; in contrast to Bennett's uncomputing, the auxiliary register is reset on-the-go without needing additional qubits. However, these improvements do not come for free: the main cost of our scheme will arise from the information to access the entanglement.

\paragraph*{List of results.} In the setting of the Abelian hidden subgroup problem, we use partial information about entanglement to optimize the erasure of auxiliary registers. In particular: 
\begin{enumerate}
    \item We find that optimal erasure (where all the entanglement between registers is exploited) is only possible if we already know the solution to the problem, i.e.\ the hidden subgroup (Theorem~\ref{thm:nogoupperboundthm}).
    \item Given partial information about the problem, we provide an optimal on-the-go erasure protocol of auxiliary registers and compute its work cost (Theorem~\ref{thm:partialsolutionworkcost}). 
    \item As an alternative to erasure, we can use that same partial information to simplify the algorithm, so that it uses fewer qubits (Theorem~\ref{thm:partialinfocorrespondence}). We provide explicit protocols for the cases of black-box oracles and open circuit access to oracles (Figure~\ref{fig:otg_hsp} and~\ref{fig:simplification_hsp}).
    \item There is a precise trade-off between the thermodynamic cost of erasure and algorithm simplification. The optimal choice of implementation (in terms of computational complexity) depends on the oracle: if we have open circuit access to the oracle, it is more efficient to simplify the circuit; if the oracle is given as a black box it is roughly equivalent to perform on-the-go erasure or to simplify the circuit. 
\end{enumerate}

\paragraph*{Structure.} In Section~\ref{sec:setting} we review the mathematical tools and notions of quantum thermodynamics along with the algorithm solving the Abelian hidden subgroup problem. These are the main ingredients on which our results are based which will be shown in Section~\ref{sec:results}.
By the example of the period finding algorithm in Section~\ref{sec:toy_example_pfa}, we illustrate the key concepts of our optimized on-the-go erasure scheme and in Section~\ref{sec:general_bounds_otg_erasure}, we generalize the example to the Abelian hidden subgroup problem. There, we state the main theorems \ref{thm:upperboundthm} - \ref{thm:partialinfocorrespondence} together with a qualitative sketch of the proofs. 
Discussions and open questions can be found in Section~\ref{sec:discussion}.
The full proofs of the main theorems and further generalizations are explored in the appendix: in appendix~\ref{appendix:erasure} an explicit erasure protocol \cite{Skrzypczyk2014} is reviewed, and appendices \ref{appendix:online} and \ref{appendix:simplification} contain the proofs for our results.

\section{Setting \& building blocks}
\label{sec:setting}
In this section we briefly review the results obtained in \cite{Rio_2011} regarding optimal bounds for the thermodynamic costs of erasing a memory with quantum side information --- this will be useful as a building block for our erasure schemes. Then we recall the algorithm solving the Abelian hidden subgroup problem, and lastly, we devise a strategy for how to optimize the erasure of the auxiliary register of said algorithm.

\subsection{Erasure with quantum side information}
\paragraph*{Work cost of erasure.} Landauer's principle \cite{landauer,Bennett2003} demonstrates the intricate relation between information theory and thermodynamics. It states that logically irreversible operations come with an intrinsic work cost, related to the temperature of the environment where the computation is carried out.
If we are looking at a system $S$ initially in a state $\rho_S$, the average work cost of erasing this system at a temperature $T$ (that is setting $\rho_S\mapsto\ket{0}\bra{0}_S$ using a thermal bath at temperature $T$) scales with the entropy of the initial state,
\begin{equation}\label{eq:landauer_principle}
    W(S)= H(S)k_BT\ln 2,
\end{equation}
where $k_B$ is the Boltzmann constant and $H(\rho)=-\tr({\rho\log_2\rho})$ is the von-Neumann entropy \cite{Goold2016}.
In the setting of Figure~\ref{fig:algorithm_otg_bruteforce}, being ignorant about the state of the $m$-qubit auxiliary system, one has to apply a fixed erasure map and not the optimal map designed for the actual state $\rho_S$. The average work cost of this map corresponds to the worst-case scenario of erasing a fully mixed state $\rho_S=(\mathds{1}/2)^{\otimes m}$, that is $mk_B T \ln 2$ for erasure at temperature $T$.
This energy is then dissipated into the rest of the quantum computer, causing it to heat up, which may increase noise and decoherence.
Using side information, available as entanglement between the main and auxiliary registers, we attempt to improve this work cost by using the following result.

\begin{restatable}[Erasure with quantum side information \cite{Rio_2011}]{lemma}{theoremlidia}
    \label{thm:theoremlidia}
    Given two degenerate quantum registers $G$ and $S$ and any reference system $R$, then there exists a process $\mathcal{E}$ acting on $G$, $S$ and an environment at temperature $T$ that erases $S$ while preserving $G$ and $R$, that is,
    \begin{equation}
        \rho_{RGS}\xmapsto[]{\mathcal{E}}\rho_{RG}\otimes \ket{0}\bra{0}_S,\text{ where }\rho_{RG}=\tr_S(\rho_{RGS}),
    \end{equation}
    which does not exceed an average work cost (and heat dissipation) of
    \begin{equation}\label{eq:landauer_principle_better}
        W(S|G)_\rho=H(S|G)_\rho \ k_B T \ln2,
    \end{equation}
    with $H(S|G)_\rho =H(GS)_\rho-H(G)_\rho$ the conditional von-Neumann entropy of $S$ conditioned on $G$. This procedure is reversible on $GS$: there exists a process that achieves the transformation $\rho_{G}\otimes \ket{0}\bra{0}_S \mapsto \rho_{GS}$ for the symmetric work cost $- W(S|G)_\rho$. 
\end{restatable}
The key insight of Lemma~\ref{thm:theoremlidia} is that quantum correlations (and in particular entanglement) can be used as additional resources to reduce the work cost of erasure of the auxiliary system. 
The average work cost is meant with respect to the thermodynamic limit of many independent copies of the systems $GS$ (for a brief discussion of single-shot and finite-size effects, see Section~\ref{sec:discussion}). In the following
$G$ will be the main register and $S$ will be the auxiliary register\footnote{The notation $G$ for the main register is chosen because later the main register will encode a group $G$.}.
The reference system $R$ includes the non-accessible degrees of freedom that may be correlated with our quantum registers, e.g.\ the rest of the quantum computer.

\paragraph*{Example: erasure of half of a Bell pair \cite{Rio_2011}.} This is the simplest application of Lemma~\ref{thm:theoremlidia}, which will be useful to understand the general procedure later. 
Take  system $S$ to be a single qubit with Hilbert space $\hilbert_S=\C^2$ and $G$ to be two qubits with $\hilbert_G=\hilbert_{G_1}\otimes\hilbert_{G_2}=\C^2\otimes\C^2$. Suppose that initially, $G_2$ and $S$ are entangled, 
\begin{equation}
    \rho_{GS}=\rho_{G_1}\otimes\ket{\chi}\bra{\chi}_{G_2S},
\end{equation}
where $\ket{\chi}=(\ket{00}+\ket{11})/\sqrt{2}$ is a fully entangled Bell state.
The goal is to erase $S$ while preserving $G$ --- that is, the final state should be $\rho_G\otimes\ket{0}\bra{0}_S=\rho_{G_1}\otimes \frac{\mathds{1}_{G_2}}2\otimes\ket{0}\bra{0}_S$. We achieve that with the following protocol:
\begin{enumerate}
    \item Unitarily rotate the pure state of $G_2S$ from $\ket{\chi}$ to $\ket{00}$ for free,
    \item Perform reverse erasure on $G_2$, to end up in state $\rho_{G_1} \otimes \frac{\mathds{1}_{G_2}}2\otimes\ket{0}\bra{0}_S$, gaining $k_BT\ln2$ work.
\end{enumerate}
This erasure map, decomposed in the following two steps,
\begin{align}
    \rho_{G_1}\otimes\ket{\chi}\bra{\chi}&\xmapsto[\text{free}]{U}\rho_{G_1}\otimes\ket{0}\bra{0}_{G_2}\ket{0}\bra{0}_S\\ &\xmapsto[\text{gain }k_BT\ln2]{\mathcal{E}^\dagger}\rho_{G_1}\otimes\frac{\mathds{1}_{G_2}}{2}\otimes\ket{0}\bra{0}_S,
\end{align}
does not affect the reduced state of the $G$ register:
\begin{align}
    \rho_G&=\tr_S\left({\rho_{GS}}\right) = \rho_{G_1}\otimes\frac{\mathds{1}_{G_2}}{2}\\
    &=\tr_S\left(\rho_{G_1}\otimes\frac{\mathds{1}_{G_2}}{2}\otimes\ket{0}\bra{0}_S\right).
\end{align}
At the end, the total average work cost of erasure for this toy example is
\begin{equation}
    W(S|G)=-k_BT\ln2=H(S|G)k_BT\ln2.
\end{equation}
in accordance with Eq. \eqref{eq:landauer_principle_better}.

\subsection{Hidden subgroup problem}
\begin{figure}
    \includegraphics[width=\columnwidth]{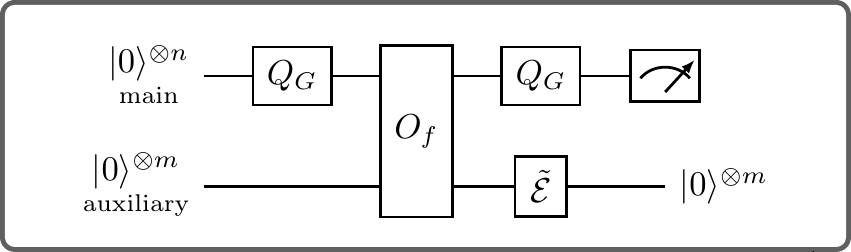}
    \caption{\textit{Abelian hidden subgroup algorithm} \cite{Mosca:1999jt,dewolf2019quantum}. The quantum circuit above solves the Abelian hidden subgroup problem. 
    Since it is of the same form as Figure~\ref{fig:algorithm_otg_optimized}, we can use it as a candidate for optimizing the on-the-go erasure.
    The main register $\hilbert_G$ encodes the group $G$ and the auxiliary register $\hilbert_S$ encodes $S$. The function oracle acts on states of the joint register via $O_f\ket{g,s}=\ket{g,s\oplus f(g)}$, with $\oplus$ denoting the bitwise XOR operation. The algorithm performs the following sequence:}\label{fig:standard_hsp_algo}
    {\begin{enumerate}
        \item Generalized quantum Fourier transform $Q_G\otimes\mathds{1}$ on $G$ register creates a superposition $\frac{1}{|G|}\sum_{g\in G}\ket{g,0}$.
        \item \label{enum:hsp_algo_2} Global oracle operation $O_f$ on both registers. At any later point we can erase $S$ with a map $\tilde{\mathcal{E}}$.
        \item Quantum Fourier transform $Q_G$ on $G$ register. 
        \item Measurement of the $G$ register.
        \item Classical post-processing of the result.
    \end{enumerate}}
\end{figure}

Several computational problems can be phrased in terms of the hidden subgroup problem (HSP) \cite{Mosca2008}, most famously period finding, which finds its application in Shor's integer factorization algorithm, and the discrete logarithm problem \cite{Shor1997}. We will first state the general problem and how our erasure algorithm applies, before looking at those particular instances.
\begin{restatable}[Hidden Subgroup Problem \cite{Mosca:1999jt}]{problem}{hsp}
\label{prob:hsp}
Let $G$ be a finite group, $S$ some finite set and $f:G\rightarrow S$ a function. Given the existence of a subgroup $H\subseteq G$ such that for all $g,g'\in G$
\begin{equation}\label{eq:defining_eq_hsp}
    f(g)=f(g')\iff gH=g'H,
\end{equation}
the goal is to determine $H$.
\end{restatable}
The HSP can be solved by an efficient\footnote{That is, polynomial time complexity under the assumption that $f$ can be implemented efficiently.} quantum algorithm originally found by \cite{kitaev1995}, under the assumption that the group $G$ is Abelian (the group operation is commutative). We will from now on be using the addition symbol $+$ for group operations in $G$ to highlight its Abelian property, that is $g+h$ instead of $gh$. Unless stated otherwise, whenever we refer to the HSP, the Abelian HSP is meant. For the general non-Abelian HSP, there are algorithms efficient in terms of oracle complexity \cite{Ettinger2004,dewolf2019quantum}, but to the authors' knowledge, no general algorithm exists that is efficient in gate complexity.
Here we follow \cite{Mosca:1999jt,dewolf2019quantum} for the quantum algorithm solving the HSP (Figure \ref{fig:standard_hsp_algo}). In Appendix \ref{appendix:hsp_group_theory} the computational steps are derived and explained in detail.
At this point, the key observation we make is that the circuit solving the HSP is precisely of the form as required, e.g.\ the circuit in Figure~\ref{fig:algorithm_vanilla}. After a unitary $U=O_f\circ\left(Q_G\otimes\mathds{1}\right)$ operation on main and auxiliary register the latter is no longer needed and can be erased by using a Landauer erasure $\tilde{\mathcal{E}}$. The computation on the main register can be continued independently.

\subsection{Strategy towards on-the-go erasure}
\label{sec:strategy}
\begin{figure}
    \includegraphics[width=\columnwidth]{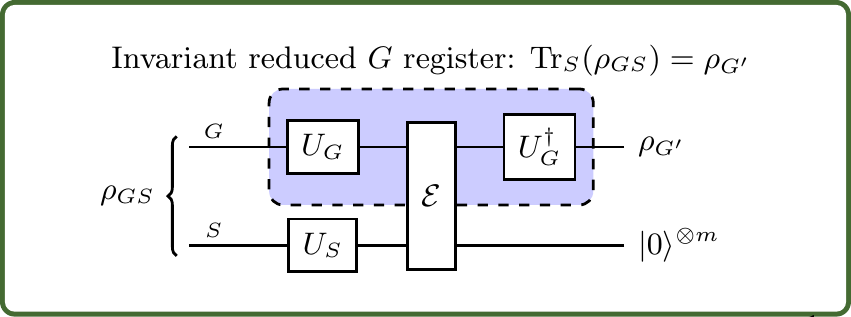}
    \caption{The erasure process acting on an initial state $\rho_{GS}$ must leave the reduced state of the main register invariant, i.e.\ we require $\tr_S(\rho_{GS})=\rho_{G'}$ (the violet box in the above Circuit must act locally as the identity on $G$). This requirement is necessary to ensure that our on-the-go erasure procedure does not affect the outcome of the algorithm to which it is applied.}
    \label{fig:InvariantGRegister}
\end{figure}

So far we have identified the point at which we optimize the erasure of the auxiliary register: Right after these qubits are not needed anymore but before the computation on the main register is finished.
The global unitary $U$ from Figure~\ref{fig:algorithm_otg_optimized} corresponds to the composition $O_f\circ\left(Q_G\otimes\mathds{1}\right)=U$ from Figure~\ref{fig:standard_hsp_algo}. By $\rho_{GS}$ we denote the state of $GS$ right after $U$.
To apply the result from Lemma~\ref{thm:theoremlidia} we have to determine \textit{where} in $\rho_{GS}$ the entanglement between $G$ and $S$ is. Operationally, this means we need to find local operations $U_G$ and $U_S$ on the main and auxiliary register respectively such that the entanglement between these registers is compressed in well-defined qubits, for example $\ell$ Bell pairs $\ket{\chi}$,
\begin{equation}\label{eq:strategy_compression}
    \rho_{GS} \xmapsto[]{U_G\otimes U_S} \rho_{{G}^{(1)}{S}^{(1)}}\otimes\left(\ket{\chi}\bra{\chi}_{G^{(2)}S^{(2)}}\right)^{\otimes \ell}.
\end{equation}
In our erasure algorithm, the entanglement is always compressed into fully entangled pairs of qubits.\footnote{Alternatively, one could weaken this assumption and consider partially mixed qubits, for which the relative entropy is greater, and therefore (by Lemma~\ref{thm:theoremlidia}), the work cost of erasure is lower.
To be applied optimally, this may require more fine-tuned control of the physical interface in the ``information battery'' part of the quantum computer (see Appendix \ref{appendix:information_battery}).}
In Section~\ref{sec:general_bounds_otg_erasure} we establish bounds on the number of Bell pairs $\ell$ for which the transformation in Eq. \eqref{eq:strategy_compression} can be achieved.
After an optimized erasure of $S$ according to Lemma~\ref{thm:theoremlidia}, the reduced state of the main register is $\rho_{{G}^{(1)}}\otimes(\mathds{1}_{G^{(2)}}/2)^{\otimes \ell}$. Before one can continue the computation with $V=Q_G$, the local transformation $U_G$ has to be undone,
\begin{equation}
    \rho_{{G^{(1)}}}\otimes\left(\frac{\mathds{1}_{G^{(2)}}}{2}\right)^{\otimes\ell}\xmapsto[]{U_G^\dagger}\rho_G',
\end{equation}
where the reduced state on $G$ is unaffected by the erasure $\rho_G'=\tr_S(\rho_{GS})$ (Figure~\ref{fig:InvariantGRegister}). This ensures that the algorithm still produces the same outcome, regardless of the manipulations due to the on-the-go erasure.
An important question we will answer in the next section is about the costs of the unitaries $U_G$ and $U_S$. While in the thermal operations resource theory they are for free, we have to quantify their cost from a computational standpoint.

\section{Results}
\label{sec:results}
Here, we introduce the optimized on-the-go erasure protocols, starting with general bounds for the HSP. Then, we will define a class of modifications realizing these optimizations whose costs we will quantify in terms of the algorithm's width. This section is concluded by a toy example for the period finding algorithm which is a special case of the HSP.

\subsection{General bounds for on-the-go erasure in the HSP}
\label{sec:general_bounds_otg_erasure}
For general unitary transformations $U_G$ and $U_S$ as sketched in the strategy from Section~\ref{sec:strategy} there is an upper bound on how well one can optimize the thermodynamics of the algorithm:
\begin{restatable}[Entanglement upper bound]{theorem}{upperboundthm}
    \label{thm:upperboundthm}
    For an Abelian group $G$ whith hidden subgroup $H$ and indicator function $f:G\rightarrow S$ as in Problem \ref{prob:hsp}, solved by the algorithm from the circuit in Figure~\ref{fig:standard_hsp_algo} the maximal number $\ell_\text{max}$ of Bell pairs between main and auxiliary registers that can be obtained via local unitary operations is 
    \begin{equation}
        \ell_{\text{max}} = \log_2\frac{|G|}{|H|} = -H(S|G)_{\rho_{GS}},
    \end{equation}
    where $\rho_{GS}$ is the state of the computational registers after the oracle operation $O_f$.
\end{restatable}

The formal proof of this statement is outsourced to Appendix \ref{appendix:proofs_hsp_erasure}. Here we sketch it: The key insight is to quantify the entanglement between $G$ and $S$ using the conditional von-Neumann entropy \cite{NJC96-QIT,Cerf1997,NF17}. As the function $f$ from the HSP is constant on cosets $g+H\in G/H$, the only entanglement that is generated by the function oracle $O_f$ comes from a sum over the \textit{different} cosets in $G/H$. Each coset $[g]\in G/H$ contributes to the entanglement by terms of the form $\ket{[g]}_{G/H}\otimes\ket{f([g])}_{S^{(2)}}$. They originate the state right after the function oracle
\begin{equation}
    \rho_{GS}=\frac{1}{|G|}\sum_{[g],[g']\in G/H}\sum_{h,h'\in H}\ket{g+h,f(g)}\bra{g'+h',f(g')}_{GS}.
\end{equation}
The sum over the cosets can be factored out via a local transformation, given by a choice of representative for each coset, that is $U_G\ket{g+h}_G=\ket{[g]}_{G/H}\otimes\ket{h}_H$. Furthermore reordering the computational basis of $\hilbert_S$ such that $\ket{f([g])}_{S^{(2)}}$ has the same computational representation as $\ket{[g]}_{G/H}$, we find
\begin{align}
    \rho_\text{rest}\otimes\sum_{[g]\in G/H}\ket{[g],f([g])}\bra{[g],f([g])}_{G/H,S^{(2)}}\\ =\rho_\text{rest}\otimes\left(\ket{\chi}\bra{\chi}\right)_{G/H,S^{(2)}}^{\otimes \log_2|G/H|}.
\end{align}
This results in a contribution of $\ell_\text{max}=\log_2|G/H|$ Bell pairs. The remaining terms in the sum are not entangled. Along the same lines we show that such a factorization can indeed be realized by unitary operations.

\begin{restatable}[Existence of transformations saturating the bound]{lemma}{upperboundexistence}
    \label{thm:upperboundexistence}
    There exist local unitaries $U_G$ and $U_S$ which saturate the upper bound $\ell_\text{max}$ of Bell pairs which can be factored from the state after the function oracle $O_f$.
\end{restatable}

There is a caveat to the transformations $U_G$ and $U_S$ saturating this upper bound as in Theorem~\ref{thm:upperboundthm}. In fact, finding the transformations must be at least as difficult as solving the problem for which we run the algorithm in the first place.

\begin{restatable}[No-go for saturating the bound]{theorem}{nogoupperboundthm}
    \label{thm:nogoupperboundthm}
    Any on-the-go erasure protocol applying local unitaries $U_G$ and $U_S$ to factorize the maximum amount $\ell_\text{max}$ of Bell pairs from Theorem~\ref{thm:upperboundthm} can be used to solve the HSP.
\end{restatable}

The underlying reason is that the transformation $U_G$ required for this factorizes the main register $\hilbert_G$ into parts belonging to $H$ and $G/H$,
\begin{equation}
    U_G : \hilbert_G\rightarrow\hilbert_H\otimes\hilbert_{G/H}.
\end{equation}
Essentially this means we have operational access to the elements of $H\subseteq G$ via the inverse operation $U_G^\dagger$. This hints at a relation between the number of Bell pairs we can factorize and the amount of information we have about the solution of our problem.
In a next step we explore how Theorem~\ref{thm:nogoupperboundthm} generalizes to instances where $\ell_\text{max}$ is not reached. What type of partial information is required to factor $\ell\leq\ell_\text{max}$ Bell pairs, and how do we quantify it? 

\subsection{On-the-go erasure and limits with partial information}
\label{sec:otg_erasure_limits}
\begin{figure*}
    \includegraphics[width=\textwidth]{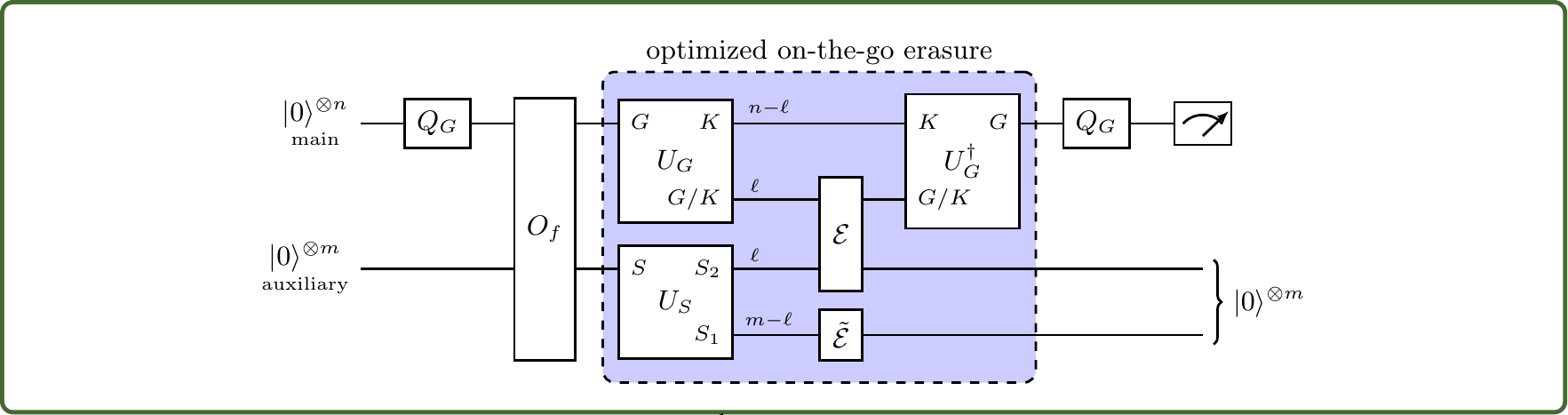}
    \caption{The above circuit is a modified version of Figure~\ref{fig:standard_hsp_algo} implementing the optimized on-the-go erasure of $\ell$ qubits (see shaded part of the diagram). Here we know unitaries $U_G$ and $U_S$ which factor out part of the entanglement between the main and auxiliary register in the form of $\ell$ Bell pairs. The $\ell$ qubits belonging to the auxiliary register are then erased at temperature $T$ with $\Tilde{\mathcal{E}}$ at a total average work cost of $-\ell k_BT\ln2$.
    The steps of the modified algorithm displayed above that are unchanged from the original are grayed out.}
    \label{fig:otg_hsp}
    \setlength{\leftmargini}{1.4cm}
    \begin{enumerate}
        {\color{gray}\item[1 - 2.] Generalized quantum Fourier transform and oracle operation.}
        \item[2a.] Unitary transformation $U_G\otimes U_S\otimes\mathds{1}_B$.
        \item[2b.] Side information erasure $\mathcal{E}$ of $\ell$ qubits, standard erasure $\Tilde{\mathcal{E}}$ of remaining $m-\ell$ qubits.
        \item[2c.] Reverse transformation $U_G^\dagger\otimes U_S^\dagger\otimes\mathds{1}_B$.
        \setcounter{enumi}{2}
        {\color{gray}\item[3 - 5.] Quantum Fourier transform, measurement and classical post-processing.}
    \end{enumerate}
\end{figure*}

\paragraph*{Partial information.} In a first step we characterize the partial information we need to know about the indicator function $f:G\rightarrow S$ such that we are able to factor $\ell\leq\ell_\text{max}$ Bell pairs after the function oracle $O_f$.
We start with the promise of \textit{knowing where} $\ell$ Bell pairs are, that is, we have access to transformations $U_G$ and $U_S$ on $\hilbert_G$ and $\hilbert_S$ which factor out $\ell$ Bell pairs after the function oracle.
The Bell pairs we consider are fully correlated qubits which tells us that the oracle $O_f$ maps some part of $\hilbert_G$ one-to-one on $\hilbert_S$. Formally, this corresponds to a factorization $\hilbert_G\cong\hilbert_{G}^{(1)}\otimes\hilbert_{G}^{(2)}$ and $\hilbert_S\cong\hilbert_{S}^{(1)}\otimes\hilbert_{S}^{(2)}$ with $O_f$ fully correlating the spaces $\hilbert_{G}^{(2)}$ and $\hilbert_{S}^{(2)}$.
This translates into a promise about algebraic properties of $f$ which characterizes what we need to know about $f$ to factor out $\ell$ Bell pairs (Promise~\ref{promise:partialinfovague}).%
\begin{restatable}[General partial information characterization, informal version]{promise}{partialinfovague}
    \label{promise:partialinfovague}
    We need to know a factorization $G\cong G^{(1)}\times G^{(2)}$ and $S\cong S^{(1)}\times S^{(2)}$ with $|G^{(2)}|=|S^{(2)}|=\ell$. Moreover, $f$ must map $G^{(2)}$ one-to-one on $S^{(2)}$.
\end{restatable}
A formalized version of this promise is given in Appendix~\ref{appendix:proofs_hsp_erasure}, Definition~\ref{def:general_local_trsf} and Theorem~\ref{thm:generaltrsfhsp}, together with a proof that Promise~\ref{promise:partialinfovague} is sufficient and necessary for factoring $\ell$ Bell pairs.
For the ease of presentation, we will present here a subclass of partial information which respects the group structure of $G$.
Partial information of this type can be understood as \textit{narrowing down} the search for the subgroup $H\subseteq G$ to a search for $H\subseteq K$ with partial information about the function oracle (see Appendix \ref{appendix:proofs_hsp_erasure} for detailed prescriptions of the transformations). In particular, the specific form in Eq.~\eqref{eq:partial_info_trsf} allows factoring out Bell pairs as outlined in our strategy (Section \ref{sec:strategy}) in Eq.~\eqref{eq:strategy_compression}. 
\begin{restatable}[Partial subgroup information]{promise}{partialsolutionpromise}
    \label{promise:partialsolutionpromise}
    We assume to have access to partial information about the indicator function $f:G\rightarrow S$. That is, we know 
    \begin{enumerate}
        \item an intermediate subgroup $K$ between $H$ and $G$ ($H\subseteq K\subseteq G$) which operationally means to have access to a unitary operation $U_G$ which factors the main register according to
        \begin{equation}
            U_G : \hilbert_G \rightarrow \hilbert_K\otimes \hilbert_{G/K}.
        \end{equation}
        \item where $f$ maps $G/K$ in $S$; operationally that means having access to a unitary $U_S$ such that
        \begin{equation}\label{eq:partial_info_trsf}
            U_G\otimes U_S \ket{g,f(g)}=\ket{k_g,[k]}\otimes \ket{\Tilde{f}(k_g),[k]}.
        \end{equation}
    \end{enumerate}
\end{restatable}

\paragraph*{Optimized on-the-go erasure with partial information.} 

The circuit in Figure~\ref{fig:otg_hsp} implements the modifications due to the transformations $U_G$ and $U_S$ from Promise~\ref{promise:partialsolutionpromise}. This brings a reduction of the work cost of erasure which we quantify in Theorem~\ref{thm:partialsolutionworkcost}.
\begin{restatable}[Work cost of erasure with partial information]{theorem}{partialsolutionworkcost}
    \label{thm:partialsolutionworkcost}
    Given the transformations $U_G$ and $U_S$ from Promise \ref{promise:partialsolutionpromise}, there exists an on-the-go erasure protocol acting on $G$, $S$ and an environment at temperature $T$, resetting the auxiliary register $S$ after $O_f$ while preserving $G$ which does not exceed an average work cost of erasure of
    \begin{equation}
        W = (m-2\ell)k_B T\ln 2,
    \end{equation}
    where $\ell = \log_2|G|/|K|$ and $m=\log_2|S|$ is the number of qubits of the auxiliary register.
\end{restatable}
This result also generalizes to partial information from Promise~\ref{promise:partialinfovague}. The only change in the circuit of Figure~\ref{fig:otg_hsp} is that the transformations $U_G$ and $U_S$ have to be replaced by their generalized versions. In Appendix~\ref{appendix:proofs_hsp_erasure}, Theorem~\ref{thm:generalpartialsolutionworkcost} generalizes Theorem~\ref{thm:partialsolutionworkcost}. For a proof, the reader is referred there.

\paragraph*{Oracle simplification with partial information.}
\begin{figure*}
    \centering
    \includegraphics[width=\textwidth]{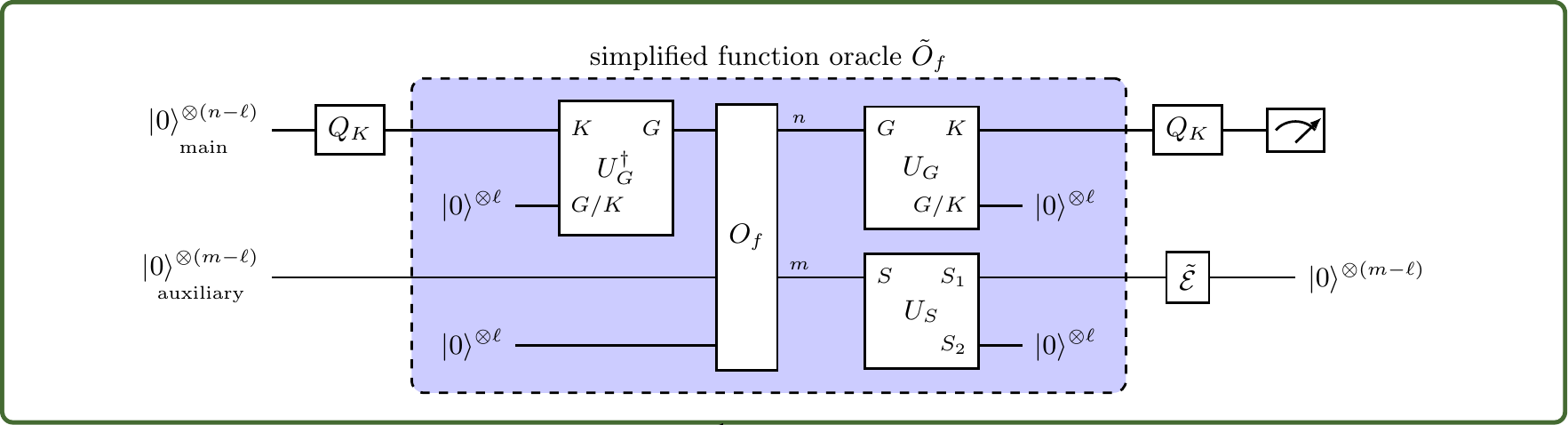}
    \caption{If we have access to partial information about a subgroup $K$, the modified oracle $\tilde{O}_f$ (violet box in the circuit above) can be used instead of the original oracle $O_f$. The $2\ell$ qubits inside the violet box are in the $\ket{0}$ state regardless of the input of the main and auxiliary qubits. All in all, the oracle $\tilde{O}_f$ has $2\ell$ fewer (variable) input qubits than $O_f$. If the function oracle $O_f$ is given with open circuit access (in contrast to a black box), the transformations $U_G$ and $U_S$ can be incorporated into $O_f$, giving a physical reduction of $2\ell$ qubits.
    In comparison to the standard algorithm (Figure~\ref{fig:standard_hsp_algo}) the group $G$ has been replaced by $K$, henceforth, also the generalized quantum Fourier transforms $Q_G$ had to be replaced by $Q_K$. The remaining steps are as in Figure~\ref{fig:standard_hsp_algo}; in the following enumeration they are grayed out, while steps 2a' - 2c' are encapsulated by $\tilde{O}_f$ in the above circuit:}
    \label{fig:simplification_hsp}
    \setlength{\leftmargini}{1.4cm}
    \begin{enumerate}
        {\color{gray}\item Generalized quantum Fourier transform $Q_K$ on $\hilbert_K$,}
        \item[2a'.] Inverse transformation $U_G^\dagger$ on $\hilbert_K\otimes\hilbert_{G/K}$,
        \item[2b'.] Original oracle operation $O_f$,
        \item[2c'.] Transformations $U_G\otimes U_S$,
        \setcounter{enumi}{2}
        {\color{gray}\item[3 - 4.] Generalized quantum Fourier transform $Q_K$ on $\hilbert_K$, and measurement of $K$ register.}
    \end{enumerate}
\end{figure*}

In Section~\ref{sec:otg_erasure_limits} we derived a no-go result (Theorem~\ref{thm:nogoupperboundthm}) for the factorization $\hilbert_G\rightarrow\hilbert_H\otimes\hilbert_{G/H}$ by observing that finding such a factorization is as difficult as finding the hidden subgroup $H\subseteq G$ itself.
With the newly introduced partial information erasure (Promise~\ref{promise:partialsolutionpromise}), how do we now quantify the difficulty of finding the transformations $U_G$ and $U_S$? Put differently: What is the operational significance of the partial information required for the transformations $U_G$ and $U_S$? The following result answers this question.
\begin{restatable}[Partial information correspondence]{theorem}{partialinfocorrespondence}
    \label{thm:partialinfocorrespondence}
    The unitaries $U_G$ and $U_S$ from Promise \ref{promise:partialsolutionpromise} can be used to formally construct a new function oracle $\tilde{O}_f$ which requires $2\ell$ fewer qubits than $O_f$ ($\ell=\log_2|G|/|K|$). Moreover, this modified oracle $\tilde{O}_f$ can still be used to solve the HSP.
\end{restatable}
With the modified oracle $\tilde{O}_f$ the HSP algorithm can be run on $2\ell$ fewer qubits than with $O_f$. Both the main and auxiliary qubits can be reduced by $\ell$ and in comparison to the circuit in Figure~\ref{fig:otg_hsp}, the quantum Fourier transform on the main register is now implemented for the group $K$ instead of $G$. Details for the proof of Theorem~\ref{thm:partialinfocorrespondence} are in Appendix \ref{appendix:simplification}.
The two constructions we made are thermodynamically equivalent: For the circuit in Figure~\ref{fig:otg_hsp} we have an average work cost of erasure equal $(m-2\ell)k_B T \ln2$ due to an erasure of $\ell$ auxiliary qubits which were fully entangled to the main register. In the simplified algorithm using modified oracle from Figure~\ref{fig:simplification_hsp}, $2\ell$ fewer qubits are required to run, hence, the average work cost of erasure is also $(m-2\ell)k_B T\ln2$.
The two constructions also produce the same computational output; in Appendix~\ref{appendix:simplification} it is shown that the construction for Theorem~\ref{thm:partialinfocorrespondence}, given in the circuit of Figure~\ref{fig:simplification_hsp} is sufficient for finding the hidden subgroup $H$.
The simplification due to the modified oracle $\tilde{O}_f$ (see Figure~\ref{fig:simplification_hsp}) can be categorized in two ways:
\begin{enumerate}
    \item $O_f$ is given as a \textbf{black box}: The simplification $\tilde{O}_f$ is a formal construction of an existence result.
    \item $O_f$ is given with \textbf{open circuit access}: The transformations $U_G$ and $U_S$ can be incorporated into the oracle $O_f$, and the new oracle requires $2\ell$ fewer physical qubits.
\end{enumerate}
\subsection{Special cases of the hidden subgroup problem}
\subsubsection{Toy example with period finding}
\label{sec:toy_example_pfa}
In this simple example, the on-the-go erasure protocol is straightforward.
The period finding algorithm (PFA) is concerned with the following problem, which is a special case of the HSP:
\begin{figure*}
    \subfloat[\label{fig:detailed_pfa_optimized} On-the-go erasure for the period finding algorithm.]{\includegraphics[width=\textwidth]{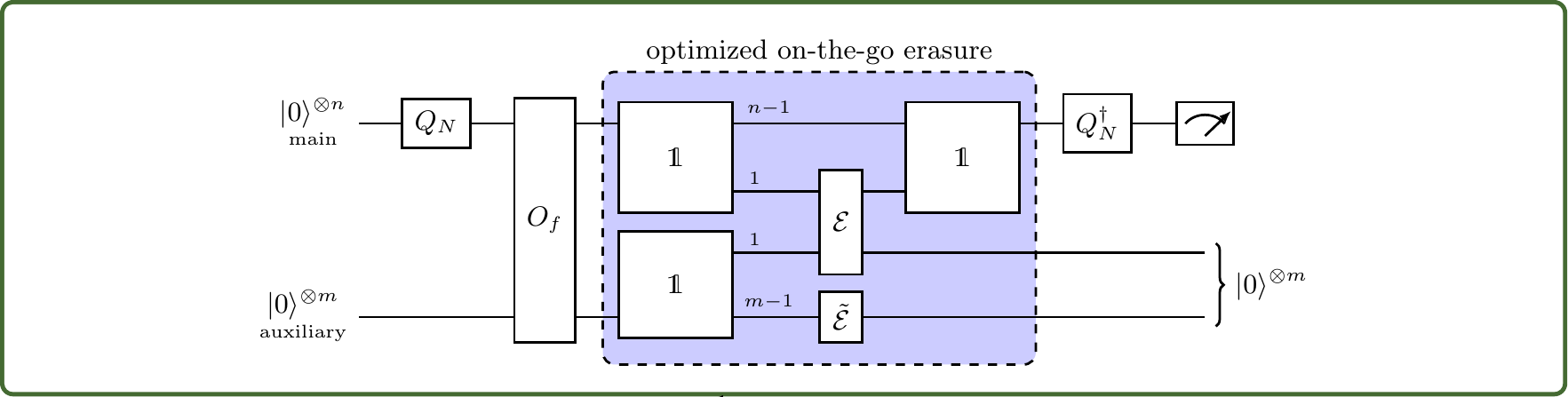}}\\
    \subfloat[\label{fig:detailed_pfa_simplified} Oracle simplification for the period finding algorithm.]{\includegraphics[width=\textwidth]{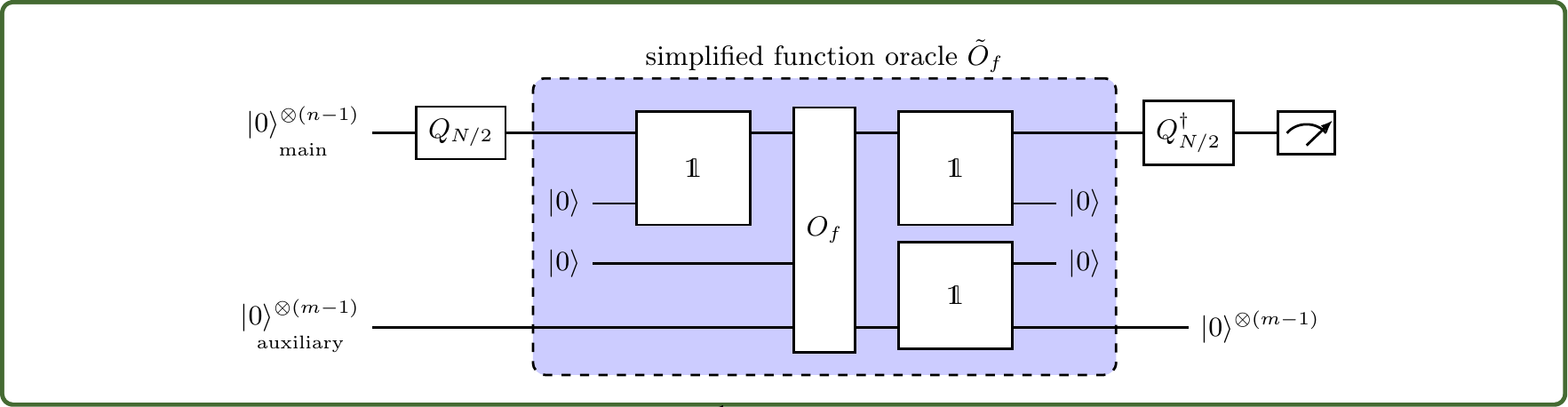}}\\
    \caption{\textit{Two-qubit optimizations to the period finding algorithm.} (\ref{fig:detailed_pfa_optimized}) Here, the circuit describing the optimized on-the-go erasure of the least significant qubits in the period finding algorithm is shown. Given the promise that the function $f$ maps even numbers to even numbers, the oracle $O_f$ fully entangles the least significant qubits of main and auxiliary register. The auxiliary qubit of this pair can then be erased at a negative average work cost of $-k_BT\ln2$ if the erasure is performed at temperature $T$. (\ref{fig:detailed_pfa_simplified}) At equivalent thermodynamic costs, the period finding algorithm can alternatively be simplified to an algorithm which uses $2$ qubits fewer by keeping the least significant qubits of main and auxiliary register constantly in the ground state $\ket{0}$.}
    \label{fig:two_qubit_erasure_pfa}
\end{figure*}

\begin{problem}[Period finding problem]
\label{problem:period_finding}
    Given a function $f:\Z_N\rightarrow\Z_M$ which is $r$-periodic and injective on each period of length $r$, the goal is to find $r$.
\end{problem}
The quantum algorithm solving this problem is of the same form as the HSP algorithm, with $G=\Z_N$, $S=\Z_M$ and $Q_G$ replaced by the standard quantum Fourier transform $Q_N$. The main register uses $n=\log_2N$ and the auxiliary register uses $m=\log_2 M$ qubits. After the first two steps the quantum state of main and auxiliary register equals
\begin{equation}\label{eq:pfa_state_after_oracle}
    \rho_{GS} = \frac{1}{N}\sum_{i,j=0}^{N-1}\ket{i,f(i)}\bra{j,f(j)}_{GS}.
\end{equation}
Suppose we were in possession of partial information about the function $f$, in form of a promise.%
\begin{promise}
    \label{promise:even_to_even}
    The function $f$ can be written in the form $f(2x(+1))=2\tilde{f}(x)(+1)$, for some other function $\tilde{f}:\Z_{N/2}\rightarrow\Z_{M/2}$. In particular, it maps even numbers to even numbers and odd to odd.
\end{promise}
This example was first proposed in \cite{Rio_2011} and it is a special case of Promise~\ref{promise:partialsolutionpromise}; here, we go through the calculations of the optimized on-the-go erasure (Figure~\ref{fig:detailed_pfa_optimized}) and  provide an explicit simplification of the function oracle (Figure~\ref{fig:detailed_pfa_simplified}).
First of all, if $f$ maps even to even numbers and odd to odd, the least significant qubits of main and auxiliary register in Eq.~\eqref{eq:pfa_state_after_oracle} are always fully entangled. By reordering them, we can write
\begin{align}
    \rho_{GS}&=\frac{1}{N/2}\left(\sum_{i,j=0}^{N/2-1}\ket{i,\tilde{f}(i)}\bra{j,\tilde{f}(j)}_{G^{(1)}S^{(1)}}\right)\\&\quad\otimes\frac{1}{2}\left(\sum_{k,\ell=0}^1\ket{kk}\bra{\ell\ell}_{G^{(2)}S^{(2)}}\right)\\
    &=\rho_{{G}^{(1)}{S}^{(1)}}\otimes\ket{\chi}\bra{\chi}_{G^{(2)}S^{(2)}},
\end{align}
and apply the result from Lemma~\ref{thm:theoremlidia} to the Bell state $\ket{\chi}$. Since the reduced main register is unaffected by this erasure, the algorithm still works to determine the period $r$. In this case, the local unitary operations $U_G$ and $U_S$ with the purpose to compress the entanglement between main and auxiliary register into well defined qubits can be chosen to be trivial, $U_G=\mathds{1}_G$ and $U_S=\mathds{1}_S$ (Figure~\ref{fig:detailed_pfa_optimized}). Ultimately, the reason for this was that (part of) the entanglement was between well-known qubits: the least significant ones. In general, however, this cannot be assumed to be the case.\par
How much worth is the partial information in Promise \ref{promise:even_to_even} in terms of computational complexity? An alternative usage of the partial information is to run the PFA not for the function $f:\Z_N\rightarrow\Z_M$ but rather $\tilde{f}:\Z_{N/2}\rightarrow\Z_{M/2}$ with $\tilde{f}(x)=f(2x)/2$. This algorithm requires $2$ fewer qubits to run. Operationally, we simply don't let the PFA act on the least significant qubits, we replace the quantum Fourier transform $Q_N$ by $Q_{N/2}$ and we let the function oracle act on all but the least significant qubits (Figure~\ref{fig:detailed_pfa_simplified}).
\subsubsection{Discrete logarithm problems}\label{sec:discrete_logarithm}
Another special case of the HSP is the discrete logarithm problem, which has applications in classical public-key cryptography.
\begin{problem}[Discrete logarithm problem]
Given the cyclic group $S=\{1,\gamma,\dots,\gamma^{N-1}\}$ of order $N$ with generator $\gamma$ and some element $A\in S$. The question is which $a\in \Z/N\Z$ satisfies $\gamma^a=A$.
\end{problem}
This problem can be rephrased as a HSP (see \cite{dewolf2019quantum} for a pedadogical derivation) by introducing the group $G=\Z/N\Z\times\Z/N\Z$ and a function
\begin{equation}
    f : G\rightarrow S;\,(i,j)\mapsto \gamma^iA^{-j}.
\end{equation}
The function $f$ is a homomorphism of groups: Let $(i,j),(k,\ell)\in G$, then
\begin{align}
    f((i,j)+(k,\ell))&=f(i+k,j+\ell)\\
    &=\gamma^{i+k}A^{-j-\ell}=\gamma^i\gamma^kA^{-j}A^{-\ell}\\
    &=f(i,j)f(k,\ell).
\end{align}
The discrete logarithm is now solved by finding the hidden subgroup $H=\langle\langle (a,1)\rangle\rangle\subseteq G$. In this formulation, the on-the-go erasure protocol is again applicable, given that partial information in the form of Promise \ref{promise:partialsolutionpromise} is available. This could again be the case in form of an intermediate subgroup $H\subseteq K\subseteq G$.

\section{Discussion}
\label{sec:discussion}
\paragraph*{Summary.} In the resource theory of thermodynamics we optimized the erasure costs of erasing auxiliary qubits in the algorithms solving the HSP. To achieve this, we applied the result from \cite{Rio_2011}, which states that quantum side-information in the form of entanglement can be used as a resource to reduce the cost of erasing quantum systems. Lastly, we quantified the cost of using said side-information in terms of a trade-off: the side-information could be used to reduce the algorithm width, at equal thermodynamic costs.

\paragraph*{Applicability.} Our work has treated three possibilities to erase auxiliary qubits in a quantum algorithm. When considering our proposal for an optimized on-the-go erasure of $\ell$ qubits for application, the following costs have to be weighted against each other. On-the-go erasure versus:
\begin{enumerate}
    \item \textbf{Straightforward erasure: } Given the architecture of the quantum computer, does the work cost reduction by $2\ell k_B T\ln2$ outweigh the gate costs of the local unitaries $U_G$ and $U_S$?
    \item \textbf{Bennett's uncomputing:} What restrictions does the quantum computer put on the algorithm's width and what is the gate cost of implementing many parallel copies of the original circuit together with a quantum version of the classical post-processing and a reversible majority vote compared to the gain of $(m-2\ell) k_B T\ln2$? Considering current gate costs or even fundamental limits \cite{CG21}, this method is unlikely to yield a thermodynamic advantage in any practical scenario.
\end{enumerate}
The toy example for the PFA demonstrates that there are cases where the local transformations $U_G$ and $U_S$ are trivial, hence they do not add any complexity to the algorithm, giving the optimized on-the-go erasure a strict advantage over approaches (1) and (2).
Last but not least, the optimized on-the-go erasure has to be compared to another option:
\begin{enumerate}
    \item[3.] \textbf{Oracle simplification:} Is $O_f$ given with open circuit access or is it given as a black box?
\end{enumerate}
If the oracle is available with open circuit access, the simplification comes with a decrease of complexity, making the algorithm use $2\ell$ fewer qubits. For a black box, the complexity is roughly the same, with the difference coming from the quantum Fourier transform which has to be performed on $\ell$ fewer qubits. At the level of thermodynamic costs, both options are equivalent.
\paragraph*{Complexity implications.} 
Depending on the type of partial information available to perform the on-the-go erasure, the complexity of the transformations $U_G$ and $U_S$ can range from being exponential in the input size to being almost trivial.
The reason for this is that Theorem~\ref{thm:partialsolutionworkcost} (together with Promise~\ref{promise:partialsolutionpromise}) is an existence result and the complexity of the transformations depends on the particular choice of computational basis representation of the states $\ket{g}$ and $\ket{s},$ for $g\in G$ and $s\in S$ respectively.
In the scenario, where we have access to side information in the form of an intermediate subgroup $K,$ such that $H\subseteq K\subseteq G,$ there is no \textit{a priori} reason for the Hilbert space $\hilbert_K=\operatorname{span}_\C\{\ket{k} : k\in K\}$ to be represented by a subregister of qubits of the main register $\hilbert_G.$ In the most general case, a unitary transformation is required to permute   basis elements and ensure $\hilbert_K$ is encoded on a subset of qubits of the main register. This transformation (which in matrix form only has $0$ and $1$ elements) has exponential gate complexity $O(n2^n)$ \cite{Shende2003}, in the general case.
This is not to say that the transformations $U_G$ and $U_S$ cannot be implemented efficiently. There are cases where the computational basis representation for $G$ already ensures that $\hilbert_K$ is implemented on a subset of the main register's qubits. In these cases, $U_G$ only has to permute qubits and has thus gate complexity bounded by $O(\log|K|)$. For example, in the PFA, this is the case for all subgroups of $G=\Z/N\Z$ generated by powers of $2$ (Section \ref{sec:toy_example_pfa}), and for the discrete logarithm for all subgroups of $G=\Z/N\Z\times \Z/N\Z$ of the form $\langle\langle2^k\rangle\rangle\times\Z/N\Z$ (Section~\ref{sec:discrete_logarithm}). For the transformation $U_S$ to satisfy similar complexity bounds, the target space's computational representation has to be decomposed analogously to the main register; this is discussed in more detail in Appendix~\ref{appendix:gate_complexity_UG_US}.

\paragraph*{Outsourcing thermodynamic processes to an information battery.} 
In this presentation, all qubit erasure processes take place in the computational registers. It is possible to outsource this thermodynamic task to an external battery register \cite{Horodecki2013,Andolina2019,2019review_lostaglio,LipkaBartosik2021secondlawof}.
The battery consists of \textit{fueled} qubits in state $\ket{0}$ and \textit{depleted} qubits in the fully mixed state $\mathds{1}/2$; the erasure of depleted qubits takes place there at temperature $T$, with an average work cost of $k_BT\ln2$ per qubit. The idea is that when we identify pure or fully mixed qubits that need erasure, we exchange them with those in the battery.
That way, all thermodynamic processes that require interaction with an environment are take place in the battery, protecting the main registers from dissipation. The price of using a battery is the need for additional SWAP gates between the computational registers and the battery, which depending on the hardware architecture may be costly.
Since the information battery does not have an effect on the number of qubits that need to be erased in a quantum computation,  further discussion is outsourced to Appendix~\ref{appendix:information_battery}.

\paragraph*{Relation to algorithmic cooling.}
Algorithmic cooling is the process of producing cold (that is, approximately pure) qubits \cite{Park2016}. There are approaches that extract entropy from a target system by coupling it to thermal baths in an approach called heat-bath algorithmic cooling \cite{AA19,Soldati2021,Serafini2020}. Our optimizations in erasure distinguish themselves from algorithmic cooling in that they are not primarily about the production of pure qubits but rather about reducing the thermodynamic cost of said erasure using entanglement as a further resource. When outsourcing the erasure process into an external information battery (see Appendix \ref{appendix:information_battery}), one could apply algorithmic cooling there to produce pure battery states.

\paragraph*{Single-shot and finite-size effects.} In our analysis, we have simplified the work cost of erasing a single qubit. Using the von Neumann entropy to quantify the work cost of erasure is an approximation valid in the asymptotic i.i.d.\ limit; for any finite number of rounds smooth entropies are more precise measures of work and heat in erasure \cite{Rio_2011,Egloff2015}. If the rest Hamiltonian of the qubits is not fully degenerate, one needs to employ single-shot versions of the free energy \cite{Aaberg2013}; if we want to account for finite-size effects (either on the environment, on thermalizing operations, or energy gaps allowed in intermediate stages of erasure), further corrections are necessary to find the exact work cost as a random variable \cite{Reeb2014,Richens2018}. All these corrections can be applied on top of our results: as mentioned in the introduction, our focus is minimizing the number of qubits that need to be erased through interaction with a thermal environment. The exact cost of that erasure can then be computed in the appropriate regime using some of the corrections above; which ones are relevant depends on the hardware architecture. Similarly, the hardware will determine the actual thermodynamic cost of individual unitary gates, which affects the calculation of whether is better to perform erasure on-the-go or to simplify the circuit. 

\paragraph*{Open questions.} 
A natural follow-up project is to study on-the-go erasure for arbitrary quantum algorithms. Within this setting, one could attempt to generalize the no-go result and the trade-off found in this paper for the HSP algorithm.
In that general setting it would also be interesting to explore automatization of the search for optimized erasure (or algorithm simplification) points, for example using entanglement detection \cite{Guehne2009,Li2021}, without affecting the state of the main register.

\addtocounter{subsection}{1}
\subsection*{Acknowledgements} 
We thank Marcus Huber for discussions and feedback, and we also thank the anonymous referees from PRA for helping us to improve the manuscript.
FM acknowledges the SEMP scholarship from Movetia for his research stay abroad and the stipend from the QUIT group at TU Vienna.
LdR acknowledges support from the Swiss
National Science Foundation through 
SNSF project No.\ $200020\_165843$, from the Quantum Center of ETH Zurich, and from the FQXi large grant RFP-CPW-2009 Consciousness in the Physical World.
This paper is the result of FM's semester project as a masters student. The technical contributions (including proofs and algorithm proposals) and first draft of the paper are his. LdR proposed the original idea, supervised the project, and revised the manuscript. 


\appendix

\addcontentsline{toc}{section}{\sc{Appendix}}

\section{Physics background: erasure and information battery}
\label{appendix:erasure}
This first appendix is dedicated to providing an explicit protocol for erasing a fully mixed qubit at the Landauer limit (Appendix \ref{appendix:erasure_one_qubit_protocol}) and to review the basics of an information battery in quantum computing (Appendix \ref{appendix:information_battery}).

\subsection{Explicit thermodynamic protocol for erasure of a fully mixed qubit}\label{appendix:erasure_one_qubit_protocol}
Erasing a fully mixed qubit, that it mapping $\mathds{1}/2\mapsto \ket{0}\bra{0}$ comes with diverging resource costs by the third law of thermodynamics \cite{Nernst1906_chem_glgw,Nernst1906_max_arb,PhysRevX.7.041033} which has been established in quantum thermodynamics as well, with diverging resource costs being time, energy or control complexity \cite{2018freitasLluis,2019clivaz,Taranto2021}. Here we showcase a protocol \cite{Skrzypczyk2014} which asymptotically implements the erasure of a qubit. The setup for the erasure consists of three quantum systems.
\paragraph*{Qubit.}
The system of the qubit is described by the two dimensional Hilbert space $\mathcal{H}_S=\C^2$ with basis $\{\ket{0}_S,\ket{1}_S\}$.
Furthermore, it is assumed that the energy levels of this system are degenerate, this is achieved with the Hamiltonian $H_S=0$.
\paragraph*{Work storage.}
In an idealized scenario the work storage consists of an infinite number of evenly spaced, non-degenerate energy eigenstates $\mathcal{H}_W=\{\ket{E_k} : k\in\mathbb{Z}\}$.
The Hamiltonian in given by $H_W=\sum_{k\in\mathbb{Z}}k\Delta\ket{E_k}\bra{E_k}$ with $\Delta$ the energy spacing between two neighbouring levels $\ket{E_k},\ket{E_{k+1}}$.
An experimental realization will only be able approximate this system with energy levels bounded from below.
Because the explicit implementation of the qubit erasure is not relevant for the remaining treatment of the online erasure, we will not investigate this any further.
\paragraph*{Heat bath.}
The heat bath is an ensemble of $N$ qubits thermalized at a temperature $\beta=1/k_BT$ where each qubit has a different energy spacing.
The Hilbert space is $\mathcal{H}_\text{bath}=(\C^2)^{\otimes N}$, with basis $\{\ket{0}_\ell,\ket{1}_\ell\}$ for the $\ell$-th factor.
The Hamiltonian governing the dynamics of the system is $H_\text{bath}=\sum_{\ell=1}^N \ell\Delta\ket{1}\bra{1}_\ell\otimes\mathds{1}_{i\neq\ell}$.
The energy $\Delta$ is the same as for the work storage.
Requiring that the qubits of the heat bath are at a temperature $\beta$ gives the thermal state of $\mathcal{H}_\text{bath}$ to be
\begin{equation*}
    \tau(\beta)=\frac{e^{-\beta H_\text{bath}}}{Z},\qquad Z=\tr\left(e^{-\beta H_\text{bath}}\right).
\end{equation*}
\paragraph*{Erasure of a qubit.}
In a first example we consider the erasure of one fully mixed qubit $\rho=\mathds{1}/2$ in $\mathcal{H}_S$.
For a heat bath consisting of $N$ qubits an erasure is performed in $N$ steps.
In step $\ell$ ($1\leq \ell\leq N$) the qubit from $\mathcal{H}_S$ is swapped with the $\ell$-th qubit from the heat bath $\mathcal{H}_\text{bath}$ and simultaneously the energy level of the work storage is lowered by $\ell$ steps to preserve energy. The unitary operation implementing this step is
\begin{widetext}
    \begin{align}
        U^{(\ell)}=&\sum_{k\in\Z}\left\{\vphantom{\sum_{i}^N}\ket{E_{k+\ell},1_S,0_\ell}\bra{E_{k},0_S,1_\ell}+\ket{E_{k},0_S,1_\ell}\bra{E_{k+\ell},1_S,0_\ell}\right\} \otimes \mathds{1}_{i\neq\ell} \label{eq:erasure_unitary_ell}\\
        &+\mathds{1}_B\otimes\Big(\ket{0_S,0_\ell}\bra{0_S,0_\ell}+\ket{1_S,1_\ell}\bra{1_S,1_\ell}\Big)\otimes\mathds{1}_{i\neq\ell},
    \end{align}

\end{widetext}
and it commutes with the Hamiltonian of the joint system of the work storage, qubit and heat bath. The energy level diagram in Figure~\ref{fig:qubit_erasure} (adapted from \cite{Skrzypczyk2014}) visualizes this unitary operation for $\ell=3$.

\begin{figure}
    \includegraphics[width=\columnwidth]{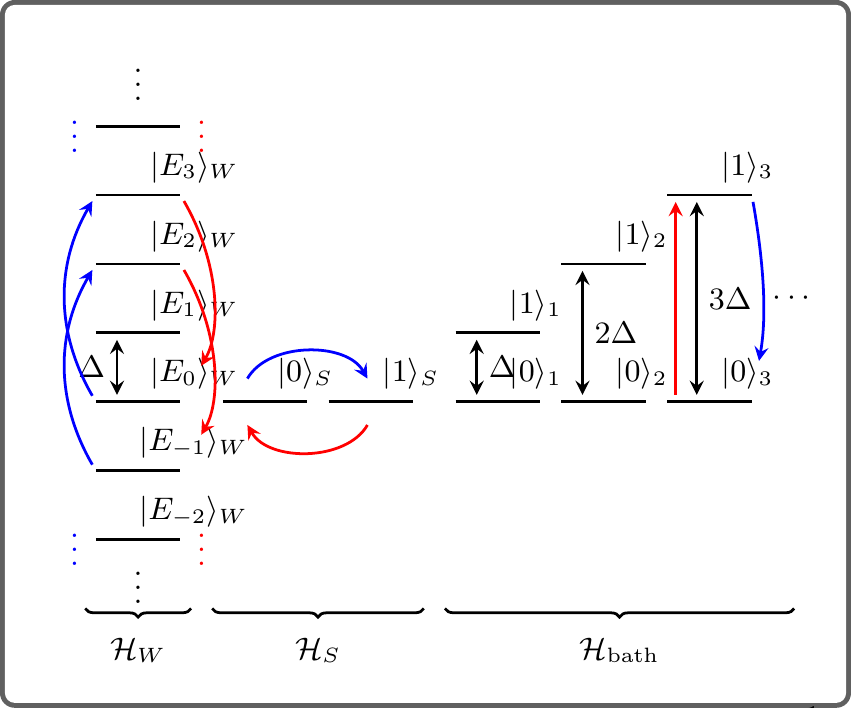}
    \caption{\textit{Qubit erasure energy diagram.} Energy level diagram for the erasure setup of one qubit. From left to right are work storage, qubit system and heat bath. The curved arrows visualize the swapping operation done by the unitary $U^{(3)}$ from Eq.~\eqref{eq:erasure_unitary_ell}.}
    \label{fig:qubit_erasure}
\end{figure}

The erasure process is the composition $U=U^{(N)}\cdots U^{(2)}U^{(1)}$. With the initial state
\begin{equation*}
    \rho_i=\ket{E_0}\bra{E_0}\otimes \frac{\mathds{1}_S}{2}\otimes \tau(\beta),
\end{equation*}
the erasure $U$ acts on $\rho_i$ such that after the erasure we are left with the reduced state of the $S$ register
\begin{equation*}
    \tr_{B,\text{bath}}\left(U\rho_i U^\dagger\right)\propto\left(\ket{0}\bra{0}_S+e^{-N\beta\Delta}\ket{1}\bra{1}_S\right).
\end{equation*}
For a large number $N$ of heat bath qubits, this process corresponds to an erasure of the qubit in system $S$: The fully mixed state $\mathds{1}_S/2$ is mapped to $\ket{0}\bra{0}_S$ asymptotically as $N\rightarrow\infty$. In this process, the work storage performs a work of $W=k_B T\log 2$ for the erasure which is dissipated as heat into the bath. In the more general case, where the system qubit is not necessarily a fully mixed state but rather $\rho_S=(1-p)\ket{0}\bra{0}_S+p\ket{1}\bra{1}_S$, the erasure unitary $U$ can be truncated which leads to a lower erasure cost of $W(S)=H(S)k_BT\log 2$ with $H(S)$ the von Neumann entropy of the system $S$. This is an explicit realization of the result from \cite{Rio_2011} for a single qubit. For many qubits this process can be performed on each qubit individually.

\subsection{Using an information battery inside the quantum computer}
\label{appendix:information_battery}
\paragraph*{Qubit battery registers.} For the purpose of this work it suffices to consider two types of battery registers: One register $\hilbert_{B}^\text{(depleted)}$, containing only fully mixed qubits
\begin{equation}
    \rho_{B}^\text{(depleted)}=\left(\frac{\mathds{1}}{2}\right)^{\otimes \dim B^\text{(depleted)}},
\end{equation}
which are completely passive \cite{Pusz1978,Skrzypczyk2015}, that is, there exists no unitary operation extracting energy from such a state. The second register $\hilbert_{B}^\text{(fueled)}$ contains only pure qubits,
\begin{equation}
    \rho_{B}^\text{(fueled)}=\left(\ket{0}\bra{0}\right)^{\otimes \dim B^\text{(fueled)}}.
\end{equation}
Using a thermal reservoir at temperature $T$ it is at best possible to extract $k_BT\ln 2$ work from a fueled qubit.
In our modifications to the HSP algorithm, instead of partially erasing Bell pairs in the computational registers, we swap them with a fully mixed and a pure qubit from the battery $\hilbert_B=\hilbert_B^\text{(depleted)}\otimes\hilbert_B^\text{(fueled)}$,

\begin{center}
    \begin{tikzcd}[nodes in empty cells,nodes={inner sep=6pt}]
        \hilbert_G\otimes \hilbert_S \arrow[mapsto,rr,bend left,"\ket{\chi}\bra{\chi}"] & [20pt] & \arrow[mapsto,ll,bend left,"(\mathds{1}/2)\otimes\ket{0}\bra{0}"] \hilbert_B,
    \end{tikzcd}
\end{center}

which amounts to a gain of one pure (fueled) qubit in the information battery. 

\paragraph*{Extraction by Swapping.}
If the entanglement between the main register $G$ and the auxiliary register $S$ is given by fully entangled qubits, for example in the Bell state $\ket{\chi}=(\ket{00}+\ket{11})/\sqrt{2}$, then this state can be replaced by fully mixed qubit for $G$ and a pure qubit for $S$ via a swapping operation. On the reduced $\hilbert_G\otimes \hilbert_S$ register, this is equivalent as a partial erasure of the qubit from $S$, while preserving $G$.
\paragraph*{General Bell pair extraction.}
\begin{figure}
    \includegraphics[width=\columnwidth]{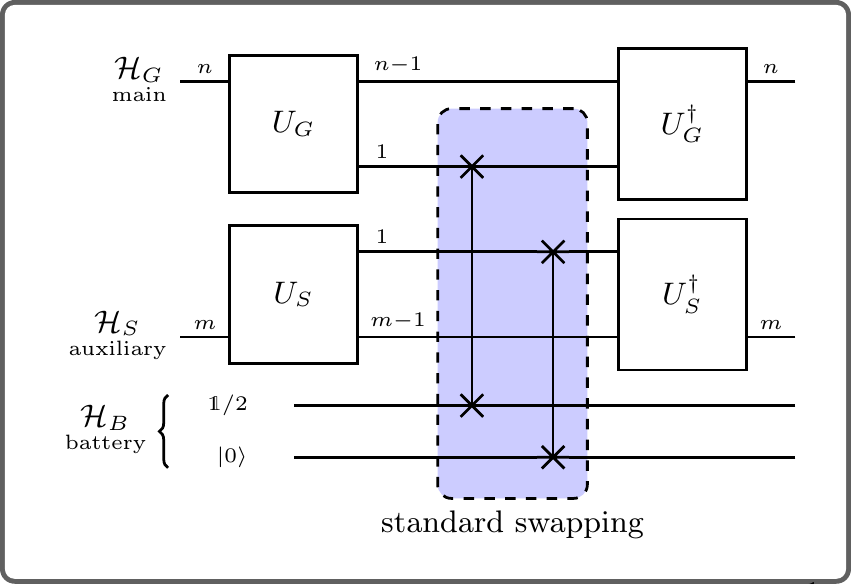}
    \caption{The first transformation brings the state $\rho_{GS}\otimes\left(\mathds{1}/{2}\right)\otimes\ket{0}\bra{0}$ into a factorized form $\rho_{G'S'}\otimes\ket{\chi}\bra{\chi}\otimes\left({\mathds{1}}/{2}\right)\otimes\ket{0}\bra{0},$ in which the Bell pair $\ket{\chi}\bra{\chi}$ is swapped with the partially erased state $\left(\mathds{1}/{2}\right) \otimes\ket{0}\bra{0}$.
    After uncomputing $U_G$ and $U_S$, that is $\rho_{G'S'}\otimes\left({\mathds{1}}/{2}\right)\otimes\ket{0}\bra{0}\otimes\ket{\chi}\bra{\chi}$ is mapped to $\tilde{\rho}_{GS}\otimes\ket{\chi}\bra{\chi},$
    the reduced states of the main register $\tilde{\rho}_G=\tr_S(\tilde{\rho}_{GS})$ is the same as before the operation, $\rho_G=\tilde{\rho}_G$, where $\rho_G=\tr_S(\rho_{GS})$.}
    \label{fig:general_bellpair_swap}
\end{figure}

In general one is not lucky enough for the entanglement between main and auxiliary register to be given in the form of well defined Bell pairs.
Since local unitary transformations of $G$ and $S$ respectively preserve the conditional entropy between these two registers, the entanglement can be spread across many qubits.
For the class of algorithms solving the HSP, we have shown that there always exist local unitaries $U_G$ and $U_S$ such that the entanglement between the registers can be compressed into Bellpairs (see Appendix \ref{appendix:proofs_hsp_erasure}).
Instead of using these unitaries to prepare the registers for being erased as in Figure~\ref{fig:otg_hsp}, they can be used to swap the entangled states with states from the battery.
In Figure~\ref{fig:general_bellpair_swap}, the situation is presented for a single Bell pair swap.
It generalizes to many Bell pairs without any complications.

\section{Proofs: On-the-go erasure in the Hidden Subgroup Problem}
\label{appendix:online}
In Appendix \ref{appendix:hsp_group_theory} we revise the group theoretic basics of the HSP, then in \ref{appendix:standard_hsp_algo} we go through a step-by-step calculation the unmodified quantum algorithm solving the HSP, lastly in \ref{appendix:proofs_hsp_erasure} we deilver the proofs for the theorems of the main body of the paper.
\subsection{Group theory of the Hidden Subgroup Problem}\label{appendix:hsp_group_theory}
In this section an online erasure protocol is constructed for the quantum algorithm \cite{dewolf2019quantum} of the \textit{Hidden Subgroup Problem} (abbr. by HSP).
\hsp*
From now on, $G$ shall be an Abelian group. The HSP for general non-Abelian groups does not yet have an efficient quantum algorithm \cite{dewolf2019quantum}. We diverge from the notation in Eq.~\eqref{eq:defining_eq_hsp} and denote by $g+h\in G$ the element in $G$ obtained by the additive group operation on $g$ and $h$ in $G$. The unit element is $0$. In particular, the cosets with respect to some subgroup $H\subseteq G$ are from now on denoted by $\bar{g}=g+H\in G/H$. We present important definitions and results from group theory \cite{Rotman2015-lv} and representation theory \cite{fulton1991representation} which will be used in the following discussion (formulation and notation of the results from \cite{Rotman2015-lv,fulton1991representation} has been adapted to the specific setting of the HSP at hand).
\begin{theorem}[Classification of finite Abelian groups \cite{Rotman2015-lv}]\label{thm:class_theorem_ab_groups}
For any finite Abelian group $G$ there exist positive integers $a_1,\dots,a_m\in\N$ such that
\begin{equation}
    G\cong \Z/a_1\Z\times\cdots\times\Z/a_m\Z
\end{equation}
and $a_1|a_2|\cdots|a_m$, where $a_i$ for all $1\leq i\leq m$ and $m$ are uniquely determined.
\end{theorem}
\begin{proposition}[\cite{fulton1991representation}]\label{prop:cyclic_char}
For each $0\leq k \leq a-1$, the function $\chi_k:\Z/a\Z\rightarrow S^1$ declared by $\chi_k(\ell)=\omega_a^{k\ell}=e^{2\pi i k\ell/a}$ is a character of the irreducible representation
\begin{equation}\label{eq:cyclic_irrep}
    \rho_k : \ell\in\Z/a\Z \mapsto \omega_a^{k\ell}\in\C^*
\end{equation}
of the cyclic group $\Z/a\Z$. In fact, these are all characters of irreducible representations of $\Z/a\Z$.
\end{proposition}
\begin{proposition}[\cite{fulton1991representation}]\label{prop:char_ab}
The characters of $G\cong \Z/a_1\Z\times\cdots\times\Z/a_m\Z$ (c.f. Theorem~\ref{thm:class_theorem_ab_groups}) are given by
\begin{equation}
    \chi_g(h)=\chi_{g_1}^{(a_1)}(h_1)\chi_{g_2}^{(a_2)}(h_2)\cdots\chi_{g_m}^{(a_m)}(h_m),
\end{equation}
where the factors are as in Proposition~\ref{prop:cyclic_char} and elements $g,h\in G$ are understood as in the decomposition of $G$ into cyclic factors. The irreducible representation $\rho_g$ having $\chi_g$ as character is given by the product $\rho_g=\rho_{g_1}\cdots\rho_{g_m}$ with the factors as in Eq.~\eqref{eq:cyclic_irrep}.
\end{proposition}
\begin{remark}
In this special Abelian case the following basic properties are satisfied by the character $\chi_g$ as introduced in Proposition~\ref{prop:char_ab}: For any $g,h\in G:\,\chi_g(h)=\chi_h(g)$ and if we take another $g'\in G$ the characters act as group homomorphisms $\chi_h(g+g')=\chi_h(g)\chi_h(g')$. However, this is only true for characters of one-dimensional irreducible representations and does not hold in general.
\end{remark}
\begin{theorem}[First orthogonality relation of characters (abelian version) \cite{fulton1991representation}]\label{thm:first_orthogonality_relation}
Let $g,g'\in G$ be elements of the Abelian Group $G$. In the space of functions $G\rightarrow\C$ the two characters $\chi_g,\chi_{g'}$ of irreducible representations of $G$ are orthonormal
\begin{equation}
    \langle \chi_g,\chi_{g'}\rangle=\frac{1}{|G|}\sum_{h\in G}\chi_g(h)\chi_{g'}(h)=\delta_{gg'},
\end{equation}
The chararacters are defined as in Proposition~\ref{prop:char_ab}.
\end{theorem}
A reformulation of the quantum Fourier transform for states representing elements of $G$ can be given in terms of characters. Consider again the cyclic group $\Z/a\Z$, pick some $k\in\Z/a\Z$ and define
\begin{equation}
    \ket{\chi_k^{(a)}}=\frac{1}{\sqrt{a}}\sum_{\ell=0}^{a-1}\omega_a^{k\ell}\ket{\ell}=\frac{1}{\sqrt{a}}\sum_{\ell=0}^{a-1}\chi_k^{(a)}(\ell)\ket{\ell}.
\end{equation}
Indeed this is the quantum Fourier transform, where we used the abbreviation $\omega_a=e^{2\pi i/a}$ for the $a$-th root of unity. A generalization is given by
\begin{definition}[Quantum Fourier transform of a group register \cite{dewolf2019quantum}]\label{def:general_char_state}
Let $G$ be a finite Abelian group with decomposition as in Theorem~\ref{thm:class_theorem_ab_groups}. For any $g\in G$ the character state $\ket{\chi_g}$ is declared by
\begin{equation}\label{eq:def_char_state}
    \ket{\chi_g}=\frac{1}{\sqrt{G}}\sum_{h\in G}\chi_g(h)\ket{h},
\end{equation}
the functions $\chi_g$, $g\in G$ as in Proposition~\ref{prop:char_ab}. 
\end{definition}
In the algorithm solving the HSP from Figure~\ref{fig:standard_hsp_algo}, states will be transformed according to the rule in Eq.~\eqref{eq:def_char_state} and certain summands will cancel out according to Theorem~\ref{thm:first_orthogonality_relation}. The following subgroup will be of particular interest to us
\begin{equation}\label{eq:def_h_perp}
    H^\perp=\{g\in G : \forall h\in H : \chi_g(h)=1\}.
\end{equation}
Elements of $H^\perp$ define functions, their characters, which allow us for probing the subgroup $H$.

\subsection{The standard algorithm solving the (Abelian) HSP}
\label{appendix:standard_hsp_algo}
Before going into the procedure of how the online erasure in the algorithm for the HSP works we explain in this section the quantum algorithm which solves the HSP (adapted from \cite{Mosca:1999jt,dewolf2019quantum}, originally solved by \cite{kitaev1995}). We will only work out the case where $|G|$ and $|S|$ are powers of $2$ in order to avoid approximations which are needed in the more general case. The $G$ register shall be made up of the first $n=\log_2|G|$ qubits, the $S$ register consists of the next $m=\log_2|S|$ qubits. For (later) notational convenience, we refer to the ground state of the $G$ register as $\ket{0}_G$ and as that of the $S$ register as $\ket{0}_S$. If clear from the context, subscripts indicating the register are dropped. We explicitly calculate the protocol from the circuit in Figure~\ref{fig:standard_hsp_algo} from the main part of the paper.
\paragraph*{Step 1.} Denote by $\rho_\ell$ the density matrix of the joint $G$ and $S$ register after iteration step $\ell$ in the HSP algorithm, in particular $\rho_0=\ket{0}\bra{0}\otimes\ket{e}\bra{e}$. The first steps of the algorithm
\begin{align}
    \rho_0&\xmapsto[]{(1)}\ket{\chi_0}\bra{\chi_0}\otimes\ket{0}\bra{0}\\
    &=\frac{1}{{|G|}}\sum_{g,g'\in G}\ket{g}\bra{g'}\otimes\ket{0}\bra{0}=\rho_1,
\end{align}
where we used $\chi_0(g)=1$ for all $g\in G$.
\paragraph*{Step 2.} Oracle Operation:
\begin{align}
    \rho_1&\xmapsto[]{(2)}\frac{1}{{|G|}}\sum_{g,g'\in G}\ket{g}\bra{g'}\otimes\ket{f(g)}\bra{f(g')}\\
    &=\frac{1}{{|G|}}\sum_{\bar{g},\bar{g}'\in G/H}\left\{\sum_{h,h'\in H}\ket{g+h}\bra{g'+h'}\right\}\\
    &\quad\otimes\ket{f(g)}\bra{f(g')}=\rho_2.\label{eq:factoring_step2}
\end{align}
In the Eq.~\eqref{eq:factoring_step2} we used that the choice of representative of $g$ in $\bar{g}=g+H\in G/H$ doesn't affect the inner summation and that $f(g)=f(g')$ if and only if $\bar{g}=\bar{g'}$. At this stage the $S$ register is traced out
\begin{equation}\label{eq:std_hsp_redstate}
    \rho_2'=\tr_S{\rho_2}=\frac{1}{{|G|}}\sum_{\bar{g}\in G/H}\left\{\sum_{h,h'\in H}\ket{g+h}\bra{g+h'}\right\}.
\end{equation}
\paragraph*{Step 3 - 4.} The remaining two steps are performed on the reduced $G$ register. The states will be denoted by a dash $\rho_\ell'$.
\begin{align}\label{eq:std_hsp_finalstate}
    \rho_2'&\xmapsto[]{(3)}\frac{1}{{|G|}}\sum_{\bar{g}\in G/H}\left\{\sum_{h,h'\in H}\ket{\chi_{g+h}}\bra{\chi_{g+h'}}\right\}\\
    &=\left(\frac{|H|}{|G|}\right)^2\sum_{\bar{g}\in G/H}\left\{\sum_{\tilde{g},\tilde{g}'\in H^\perp}\chi_g(\tilde{g})\chi_g(\tilde{g}')\ket{\tilde{g}}\bra{\tilde{g}'}\right\}=\rho_3'.
\end{align}
The measurement result of the $G$ register gives an element $\tilde{g}\in H^\perp$. This element defines the function $\chi_{\tilde{g}}:G\rightarrow S^1$ whose restriction to $H$ is the unit function. For all $h\in H$, $\chi_{\tilde{g}}(h)=1$. Multiple iterations of the HSP algorithm give a set of such functions $\{\chi_{\tilde{g}}\}_{\tilde{g}}$ constraining $H\subseteq G$ and thus solving the problem. 

\subsection{Proofs}
\label{appendix:proofs_hsp_erasure}
\upperboundthm*
\begin{proof}
    The measure we use to quantify the degree of entanglement between the $G$ and $S$ register is the conditional von Neumann entropy
    \begin{equation}
        H(S|G)=H(\rho_{GS})-H(\rho_G),
    \end{equation}
    where $H$ is the standard von Neumann entropy \cite{NJC96-QIT,Cerf1997,NF17}. The number of Bell pairs formed by qubits from $G$ and $S$ will be upper bounded by $-H(S|G)$ as a single bell pair contributes a negative conditional entropy of $-1$. The joint state of $\mathcal{H}_G\otimes\mathcal{H}_S$ is $\rho_{GS}$ and $\rho_G=\tr_S\rho_{GS}$ is the reduced state of the $G$ register. Observe the following two facts: Firstly, up until the erasure of the $S$ register which takes place after step \ref{enum:hsp_algo_2} in the HSP algorithm, the joint state $\rho_{GS}$ is a pure state. Secondly, we have $\rho_{GS}=\rho_G\otimes\rho_S$ before step \ref{enum:hsp_algo_2} where the function oracle is applied, where $\rho_G$ and $\rho_S=\ket{0}\bra{0}$ are pure states. Assuming that the function oracle $O_f$ is a black box, the only stage of the HSP algorithm where $H(S|G)$ is non-zero is after $O_f$ but before the $\mathcal{H}_S$ is traced out. The corresponding state from Eq.~\eqref{eq:factoring_step2} is
    \begin{align}
        \rho_{GS}=\frac{1}{{|G|}}\sum_{\substack{\bar{g},\bar{g}'\in\\ G/H}}\left\{\sum_{\substack{h,h'\in \\ H}}\ket{g+h}\bra{g'+h'}\right\}\otimes\ket{f(g)}\bra{f(g')},
    \end{align}
    whose conditional entropy $H(S|G)$ is given by
    \begin{align}
        H(S|G) &= H(\rho_{GS})-H(\rho_G)\\&=0+\tr(\rho_G\log_2\rho_G),
    \end{align}
    with $\rho_G=\rho_2'$ the reduced $G$ register state, c.f. Eq.~\eqref{eq:std_hsp_redstate}. The entropy measure is invariant under unitary transformations of $\rho_G$. By reordering the computational basis of the $G$ register we can split $\mathcal{H}_G\cong\mathcal{H}_H\otimes\mathcal{H}_{G/H}$. The state $\rho_2$ can be factored $\rho_2\cong\rho_H\otimes\rho_{G/H}$
    \begin{align}
        &\frac{1}{{|G|}}\sum_{\bar{g}\in G/H}\left\{\sum_{h,h'\in H}\ket{g+h}\bra{g+h'}\right\}\\ &\quad \cong\underbrace{\frac{1}{|H|}\sum_{h,h'\in H}\ket{h}\bra{h'}}_{=\rho_H}\otimes\underbrace{\frac{|H|}{|G|}\sum_{\bar{g}\in G/H}\ket{g}\bra{g}}_{=\rho_{G/H}}.
    \end{align}
    The factor $\rho_H$ is a pure state of $\mathcal{H}_H$, the right one $\rho_{G/H}$ is a mixed state in $\mathcal{H}_{G/H}$, already represented in it's diagonal basis with eigenvalues $|H|/|G|$. The entropy of $\rho_G=\rho_2$ therefore is
    \begin{align}
        H(\rho_G)&=-\tr_H\left(\rho_H\log_2\rho_H\right)-\tr_{G/H}\left(\rho_{G/H}\log_2\rho_{G/H}\right)\\
        &=-\sum_{G/H}\frac{|H|}{|G|}\log_2\left(\frac{|H|}{|G|}\right)\\
        &=+\log_2\left(\frac{|G|}{|H|}\right),
    \end{align}
    which gives the entropy of the $S$ register conditioned on $G$
    \begin{equation}\label{eq:conditional_entropy_max}
        H(S|G)=-\log_2\left(\frac{|G|}{|H|}\right).
    \end{equation}
    Up to the negative sign this is equal the maximal number $k:=\log_2(|G|/|H|)$ of Bell pairs formed by qubits from the $G$ and $S$ register which can possibly be extracted from the HSP algorithm.
\end{proof}

\upperboundexistence*
\begin{proof}
    We argue why operations $U_G$ and $U_S$ exist such that the result from Eq.~\eqref{eq:strategy_compression} can be achieved. First, consider the $G$ register: If a choice of representative is made for each coset $\bar{g}\in G/H$, any element $g'\in G$ can be split as $g'=g+h$ with $g$ the representative of $\bar{g}'$ and $h\in H$. Thus, there exists an invertible map on $\mathcal{H}_G$ such that $\ket{g'}\mapsto\ket{h}\otimes\ket{g}\in \mathcal{H}_H\otimes\mathcal{H}_{G/H}$. The states $\{\ket{g}\}_{g\in G}$ are orthonormal, thus, this map is unitary, we shall denote it by $U_G$. In comparison to the notation in step (3) from the paragraph about the classification of all qubit extraction procedures, we have $\mathcal{H}_{G/H}=\mathcal{H}_{B_1}$, $\mathcal{H}_H=\mathcal{H}_G^{(1)}$ and
    \begin{align}
        U_G:\mathcal{H}_G&\longrightarrow{}\mathcal{H}_H\otimes\mathcal{H}_{G/H}\\
        \ket{g'}=\ket{g+h}&\longmapsto\ket{h}\otimes\ket{g}.
    \end{align}
    The relevant states in the ancillary register $H_S$ are of the form $\ket{f(g)}$ for $g\in G$. In fact, by the very defining assumption for the HSP in Eq.~\eqref{eq:defining_eq_hsp}, it suffices to restrict to representatives $g$ of cosets $\bar{g}\in G/H$. Generally, the ancillary register $\mathcal{H}_S$ may have more qubits than are actually needed to represent $\im f\subseteq S$. This overhead of qubits can be factored out by reordering the computational basis of $\mathcal{H}_S$ such that $\ket{f(g)}\mapsto \ket{0}\otimes\ket{\tilde{f}(g)}\in\mathcal{H}_S^{(1)}\otimes\mathcal{H}_{B_2}$. This operation is unitary and can be chosen such that for all representatives $g$, the states $\ket{\tilde{f}(g)}\in\mathcal{H}_{B_2}$ and $\ket{g}\in\mathcal{H}_{G/H}$ have the same computational representation. This transformation will be denoted by $U_S$.
\end{proof}

\nogoupperboundthm*
\begin{proof}
    Any local unitary $U_G$ factoring the main register register $\hilbert_G$ into $\hilbert_H\otimes\hilbert_{G/H}$ can in fact be used to determine $H\subseteq G$. Elements in $H$ can be obtained by applying the inverse $U_G^\dagger$ to states in $\hilbert_H\otimes\hilbert_{G/H}$. Pick some $g\in G$, then $U_G$ factors the state $\ket{g}$ into two parts
    \begin{equation}
        U_G\ket{g}=\ket{h_g}\otimes\ket{[g]}.
    \end{equation}
    Despite our ignorance about how the group structure is binarily encoded in the quantum registers $\hilbert_H$ and $\hilbert_{G/H}$, we know that for any $h\in H$, $\ket{[g]}=\ket{[g+h]}$. Elements from $H$ can then be obtained in two steps
    \begin{enumerate}
        \item Determine $\ket{[0]}\in\hilbert_{G/H}$ by computing $U_G\ket{0}=\ket{h_0}\otimes\ket{[0]}$.
        \item Pick any $\ket{h}_H\in\hilbert_H$ and deduce $\ket{h}_G\in\hilbert_G$ via
        \begin{equation}
            U_G^\dagger\ket{h}\otimes\ket{[0]}=\ket{h}\in\hilbert_G.
        \end{equation}
    \end{enumerate}
    In the second step the register $H$ and $G$ is highlighted for the states $\ket{h}_H$ and $\ket{h}_G$.
    That is because for $\hilbert_G$, we have access to an encoding $g\in G\mapsto \ket{g}\in\hilbert_G$ while for $\hilbert_H$ we do not.
    That is also the reason why one has to use the inverse operation $U_G^\dagger$ to obtain $H$.
    As with the functions $\chi_{\tilde{g}}$ from the standard algorithm solving the HSP in Appendix \ref{appendix:standard_hsp_algo}, this procedure can be used to determine a small number of elements $h\in H$ which then generate the whole subgroup $H$.
\end{proof}
In fact one can even go further: Finding an on-the-go erasure procedure in the setting of Theorem~\ref{thm:nogoupperboundthm} is more difficult than solving the HSP, for that it also requires the transformation $U_S$. For partial information erasure procedures we give a quantitative description of \textit{how much} information is required to compress the entanglement for an on-the-go erasure.

\begin{definition}[General local transformations of $G$ and $S$]
    \label{def:general_local_trsf}
    Define local transformations 
    \begin{align}
        U_G : \hilbert_G &\rightarrow \hilbert_{G}^{(1)}\otimes \hilbert_{G}^{(2)},\\
        U_S : \hilbert_{S} &\rightarrow \hilbert_{S}^{(1)}\otimes \hilbert_{S}^{(2)},
    \end{align}
    which factor quantum states encoding elements in $g\in G$ and $s\in S$ according to
    \begin{align}
        U_G\ket{g}=\ket{g^{(1)},g^{(2)}},\\ U_S\ket{s}=\ket{s^{(1)},s^{(2)}}.\label{eq:splitting_aux_state}
    \end{align}
\end{definition}

Similarly to the notation introduced in Eq.\ \eqref{eq:splitting_aux_state}, let us write for some state $\ket{f(g)}\in\hilbert_S$, $U_S\ket{f(g)}=\ket{f^{(1)}(g),f^{(2)}(g)}$. Using this notation we can formulate two general conditions on transformations $U_G$ and $U_S$:

\begin{restatable}[General characterization of partial erasure transformations]{theorem}{generaltrsfhsp}
    \label{thm:generaltrsfhsp}
    If and only if the transformations $U_G$ and $U_S$ satisfy the two requirements
    \begin{enumerate}
        \item\label{enum:reduced_trsf_property1} For all $g,\tilde{g}\in G:$ $f(g)=f(\tilde{g})\rightarrow g^{(2)}=\tilde{g}^{(2)}$,
        \item\label{enum:reduced_trsf_property2} The function $f^{(1)}(g)$ only depends on $g^{(1)}$ and $f^{(2)}(g)=g^{(2)}$ in the binary computational representation in $\mathcal{H}_S^{(2)}=(\C^2)^{\otimes k}=\mathcal{H}_G^{(2)}$,
    \end{enumerate}
    they can factor out the entanglement in the form of $\ell$ Bell pairs after $O_f$ in the HSP algorithm, where $\ell=\dim\hilbert_{G}^{(2)}=\dim\hilbert_{S}^{(2)}$.
\end{restatable}
\begin{proof}
    We obtain conditions on transformations $U_G$ and $U_S$ which allow bringing the joint state of the $G$ and $S$ register into the form 
    \begin{equation}
        \rho_{GS}\mapsto \rho_{\tilde{G}\tilde{S}}\otimes (\ket{\chi}\bra{\chi})^{\otimes \ell}.
    \end{equation}
    We allow that $U_G$ may only factor part of the register $\mathcal{H}_{G/H}$, say $k\leq|G|/|H|$ qubits. The transformations need not necessarily respect the group structure of $G$, hence, we refrain from using an intermediate subgroup $K$ as in the main part of the paper but rather work with the factorization from Definition~\ref{def:general_local_trsf}. Starting with the state $\rho_2$ after step \ref{enum:hsp_algo_2} of the standard HSP algorithm (see Figure~\ref{fig:standard_hsp_algo} and Appendix \ref{appendix:standard_hsp_algo}), we find a new state $\rho_{2a}$
\begin{widetext}
    \begin{align}
        \rho_{2a}&=\left(\vphantom{U_G^\dagger}U_G\otimes U_S\right)\rho_2\left(U_G^\dagger\otimes U_S^\dagger\right)\\
        &=\frac{1}{|G|}\sum_{g,\tilde{g}\in G}\ket{g^{(1)},g^{(2)}}\bra{\tilde{g}^{(1)},\tilde{g}^{(2)}}\otimes \ket{f^{(1)}(g),f^{(2)}(g)}\bra{f^{(1)}(\tilde{g}),f^{(2)}(\tilde{g})}\\
        &\stackrel{\heartsuit}{=}\frac{1}{|G|}\left\{\sum_{g^{(1)},\tilde{g}^{(1)}}\ket{g^{(1)}}\bra{\tilde{g}^{(1)}}\otimes \ket{f^{(1)}(g^{(1)})}\bra{f^{(1)}(\tilde{g}^{(1)})}\right\}\otimes\\
        &\quad\qquad\left\{\sum_{g^{(2)},\tilde{g}^{(2)}}\ket{g^{(2)}}\bra{\tilde{g}^{(2)}}\otimes\ket{f^{(2)}(g^{(2)})}\bra{f^{(2)}(\tilde{g}^{(2)})}\right\}\\
        &=\frac{1}{\dim\mathcal{H}_G^{(1)}}\left\{\sum_{g^{(1)},\tilde{g}^{(1)}}\ket{g^{(1)}}\bra{\tilde{g}^{(1)}}\otimes \ket{f^{(1)}(g^{(1)})}\bra{f^{(1)}(\tilde{g}^{(1)})}\right\}\otimes\left(\ket{\chi}\bra{\chi}\right)^{\otimes k}.\label{eq:general_hsp_trsf_state}
    \end{align}
\end{widetext}
The necessary and sufficient condition for the transformations $U_G$ and $U_S$ factoring $\ell$ Bell pairs is the equality $\heartsuit$ above.
\end{proof}

\begin{remark}
The transformations $U_G$ and $U_S$ from Promise \ref{promise:partialsolutionpromise} where one has partial information on an intermediate subgroup $H\subseteq K\subseteq G$ are a special case of the transformations from Theorem~\ref{thm:generaltrsfhsp} with
\begin{equation}
    \hilbert_G^{(1)}=\hilbert_{K},\text{ and } \hilbert_G^{(2)}=\hilbert_{G/K},
\end{equation}
as this decomposition satisfies all two assumptions from Theorem~\ref{thm:generaltrsfhsp}.
\end{remark}

\begin{restatable}[Work cost of erasure with partial information, general version]{theorem}{generalpartialsolutionworkcost}
    \label{thm:generalpartialsolutionworkcost}
    Given the transformations $U_G$ and $U_S$ from Definition~\ref{def:general_local_trsf} and Theorem~\ref{thm:generaltrsfhsp}, there exists an on-the-go erasure protocol acting on $G$, $S$ and an environment at temperature $T$, resetting the auxiliary register $S$ after $O_f$ while preserving $G$ which does not exceed an average work cost of erasure of
    \begin{equation}
        W = (m-2\ell)k_B T\ln 2,
    \end{equation}
    where $\ell = \log_2(\dim\hilbert_G^{(2)})$.
\end{restatable}
\begin{proof}
    This result follows from the form of the state in Eq.\ \ref{eq:general_hsp_trsf_state} and Theorem~\ref{thm:theoremlidia}. For completeness, the resulting of state of the circuit from Figure~\ref{fig:otg_hsp} is reproduced here in order to show it coincides with the one from the standard HSP algorithm in Figure~\ref{fig:standard_hsp_algo}.
    \paragraph*{Steps 1 - 2.} These steps are the same as for the standard HSP algorithm. The resulting state is
    \begin{equation}
        \rho_2 = \frac{1}{{|G|}}\sum_{\substack{\bar{g},\bar{g}'\in\\ G/H}}\left\{\sum_{\substack{h,h'\in \\ H}}\ket{g+h}\bra{g'+h'}\right\}\otimes\ket{f(g)}\bra{f(g')}.
    \end{equation}
    \paragraph*{Steps 2a - 2c. } Applying the operation $U_G\otimes U_S$ to the state $\rho_2$ gives (see calculation in proof of Theorem~\ref{thm:generaltrsfhsp})
    \begin{widetext}
        \begin{align}
            \rho_{2a}&=\left(\vphantom{U_G^\dagger}U_G\otimes U_S\right)\rho_2\left(U_G^\dagger\otimes U_S^\dagger\right)=\frac{1}{\dim\mathcal{H}_G^{(1)}}\left\{\sum_{g^{(1)},\tilde{g}^{(1)}}\ket{g^{(1)}}\bra{\tilde{g}^{(1)}}\otimes \ket{f^{(1)}(g^{(1)})}\bra{f^{(1)}(\tilde{g}^{(1)})}\right\}\otimes\left(\ket{\chi}\bra{\chi}\right)^{\otimes k}.\label{eq:swapping_rho2a}
        \end{align}
    \end{widetext}
    The Bell pairs $\ket{\chi}$ are formed between qubits from $\hilbert_{G}^{(2)}$ and $\hilbert_{S}^{(2)}$.
    The erasure $\tilde{\mathcal{E}}$ in step 2b of $\hilbert_S^{(1)}$ is a standard Landauer erasure at temperature $T$. We are ignorant about the state in $\hilbert_S^{(1)}$, thus we have to pay the full cost of
    \begin{equation}
        W^{(1)}=\dim(\hilbert_S^{(1)})k_B T\ln 2 = (m-\ell)k_B T\ln 2.
    \end{equation}
    Conversely the erasure $\mathcal{E}$ is done with quantum side information according to Theorem~\ref{thm:theoremlidia}. The average work cost of erasure at temperature $T$ is given by
    \begin{equation}
        W^{(2)}=H(S^{(2)}|G^{(2)})k_B T\ln2=-\ell k_B T\ln2,
    \end{equation}
    which amounts to a total average work cost of erasure
    \begin{equation}
        W=W^{(1)}+W^{(2)}=(m-2\ell)k_B T\ln2.
    \end{equation}
    The erasure leaves the reduced state of the $G$ register invariant. After uncomputing $U_G$, we get
    \begin{equation}
        \rho'_{2c}=\frac{1}{{|G|}}\sum_{\bar{g}\in G/H}\left\{\sum_{h,h'\in H}\ket{g+h}\bra{g+h'}\right\},
    \end{equation}
    as in Eq.\ \eqref{eq:std_hsp_redstate} from the standard HSP algorithm.
    \paragraph*{Steps 3 - 6.} Based on the last observation, these steps go through as for the standard case.
\end{proof}

\subsection{Gate complexity of $U_G$ and $U_S$}
\label{appendix:gate_complexity_UG_US}
The transformations $U_G$ and $U_S$ from Definition~\ref{def:general_local_trsf} which are  used in Theorem~\ref{thm:generaltrsfhsp} are permutations of the basis states $\ket{g}$ and $\ket{s}$ for $g\in G$ and $s\in S$. These permutations ensure that after the application of the function oracle $O_f,$ the entanglement is compressed into a well-defined subregister of the main and auxiliary register.

In general, a permutation unitary on the computational basis states of $n$ qubits requires $O(n 2^n)$ CNOT gates \cite{Shende2003} and is therefore not efficiently implementable.
Nevertheless, depending on the type of partial information available, the complexity of the transformations $U_G$ and $U_S$ can be drastically reduced (see for example the PFA, Section~\ref{sec:toy_example_pfa}).
To this end, let us work in the special setting where the partial information is available in the form of an intermediate subgroup $K,$ such that $H\subseteq K\subseteq G$ (as in Promise~\ref{promise:partialsolutionpromise}).
There, the transformations $U_G$ and $U_S$ act on a state $\ket{g,f(g)}$ as follows,
\begin{align}
    U_G\otimes U_S \ket{g,f(g)}=\ket{k_g,[k]}\otimes \ket{\Tilde{f}(k_g),[k]},
\end{align}
where $k_g\in K$ and $[k]\in G/K$ are a decomposition of $g\in G$ into an element in $K$ and the quotient group $G/K$.
Consider the special case where $K$ is already implemented on a subset of qubits of the main register --- that is $\hilbert_K=\operatorname{span}_\C\{\ket{k} : k\in K\}$ is the Hilbert space generated by some but not necessarily all qubits that span $\hilbert_G$. Here the transformation $U_G$ is only a composition of qubit swaps which can be implemented efficiently with a complexity $O(\log|K|)$.
For the target space an analogous rule holds. If the map $f:G\rightarrow S$ implemented on the level of the function oracle $O_f$ respects the qubit decomposition of $\hilbert_G$ into $\hilbert_K\otimes\hilbert_{G/K},$ that is, these subregisters are mapped to subregisters of the auxiliary space $\hilbert_S$, then also $U_S$ has complexity $O(\log|K|)$. One particular case where this happens is the toy example for the PFA shown in Section~\ref{sec:toy_example_pfa}.

\section{Proofs: Oracle simplification in the Hidden Subgroup Problem}
\label{appendix:simplification}
This appendix is dedicated to proving Theorem~\ref{thm:partialinfocorrespondence} and giving more details on the modified HSP algorithm using a simplified oracle. By replacing the function oracle $O_f$ by $\Tilde{O}_f$ one also has to reconsider what group the main register encodes. In fact, as the transformation $U_G$ now hidden in $\Tilde{O}_f$ factors $\hilbert_G$ into registers $\hilbert_K$ and $\hilbert_{G/K}$ encoding the groups $K$ and $G/K$ respectively, the main register now encodes the subgroup $K$. This coincides with the statement, that with the partial information, we can narrow down the search for $H\in G$ to a search of $H\in K$. Consequently, also the generalized quantum Fourier transform $Q_G$ has to be replaced by $Q_K$ as is shown in Figure~\ref{fig:simplification_hsp} from the main part of the paper with the simplified algorithm.

\partialinfocorrespondence*

\begin{proof}
    An explicit calculation of the state $\rho_i$ for the algorithm in the circuit of Figure~\ref{fig:simplification_hsp} is performed, with $1\leq i\leq 6$ indexing the steps defined there.
    \paragraph*{Steps 1 - 2. } The sum $\sum_{k}$ is implicitly over the range of the first factor in $\{\ket{k,t}=U_G\ket{g} : g\in G\}\subseteq\mathcal{H}_K\otimes\mathcal{H}_{G/K}$. Then,
    \begin{align}
        \rho_1&=\frac{1}{|K|}\sum_{k,\tilde{k}}\ket{k,0}\bra{\tilde{k},0}\otimes\ket{0}\bra{0}\\
        &\xmapsto[]{U_G^\dagger\otimes\mathds{1}_S}\frac{1}{|K|}\sum_{k,\tilde{k}}\underbrace{U_G^\dagger\ket{k,0}}_{=:\ket{g(k,0)}}\bra{\tilde{k},0}U_G\otimes\ket{0}\bra{0}=\rho_2\label{eq:notation_intro}.
    \end{align}
    \paragraph*{Steps 3 - 4. } Making use of the notation introduced in Eq.~\eqref{eq:notation_intro} where $g(k,t)\in G$ is the unique element s.t.\ $U_G\ket{g(k,t)}=\ket{k,t}$ the next states can be written down,
    \begin{align}
        \rho_2&\xmapsto[]{{O}_f}\frac{1}{|K|}\sum_{k,\tilde{k}}U_G^\dagger\ket{k,0}\bra{\tilde{k},0}U_G\otimes\ket{f(g(k,0))}\bra{f(g(\tilde{k},0))}\\
        &\xmapsto[]{U_G\otimes U_S}\frac{1}{|K|}\sum_{k,\tilde{k}}\ket{k,0}\bra{\tilde{k},0}\otimes\ket{f^{(1)}(k),0}\bra{f^{(1)}(\tilde{k}),0}=\rho_4.
    \end{align}
    For the last equality we used property \ref{enum:reduced_trsf_property2} imposed in Theorem~\ref{thm:generaltrsfhsp}. At this stage we see that $\tilde{O}_f$ acts trivially on the registers $\mathcal{H}_G^{(2)}=\mathcal{H}_{G/K}$ and $\mathcal{H}_S^{(2)}$, saving $2\ell$ erasures like the online erasure protocols do.
    
    \paragraph*{Steps 5 - 6: Recovering the hidden subgroup. } It remains to be shown that the modified algorithm can still be used to determine the hidden subgroup $H$. Let
    \begin{align}
        \rho_{4}'&=\tr_{\mathcal{H}_{G/K}\otimes\mathcal{H}_S}\rho_4\\
        &=\frac{1}{|K|}\sum_{\bar{k}\in K/H}\left\{\sum_{h,h'\in H}\ket{k+h}\bra{k+h'}\right\}
    \end{align}
    be the reduced state of $\rho_4$ where all registers but $\mathcal{H}_K$ have been traced out. Observing $H\subseteq K$, the quantum Fourier transform $Q_{K}$ acts on $\rho_4'$ as follows
    \begin{align}
       \rho_4'&\xmapsto[]{Q_{K}}\frac{1}{|K|}\sum_{\bar{k}\in K/H}\left\{\sum_{h,h'\in H}\ket{\chi_{k+h}}\bra{\chi_{k+h'}}\right\}\\\
       &=\left(\frac{|H|}{|K|}\right)^2\sum_{\bar{k}\in K/H}\left\{\sum_{g,g'\in H_K^\perp}\chi_k(g)\chi_k(g')\ket{g}\bra{g'}\right\},
    \end{align}
    where $H_{K^\perp}=\{g\in K : \forall h\in H :\chi_g(h)=1\}$ is the analogue of $H^\perp$ from the standard HSP algorithm with the difference that $G$ has been replaced by $K$. This calculation demonstrates that the modified algorithm still recovers the hidden subgroup $H$.
\end{proof}

%

\end{document}